\documentclass[a4paper]{article}
\usepackage{microtype} 
\usepackage{pdfpages}
\usepackage[T1]{fontenc}
\usepackage{mathtools,amsmath,xspace}
\usepackage{amssymb}
\usepackage{graphicx}
\usepackage{hyperref} 
\graphicspath{{figures/}}
\usepackage[nodayofweek]{datetime} 
\usepackage{xspace}
\usepackage{soul}
\usepackage[]{algorithm2e}
\usepackage{fullpage}

\newtheorem{theorem}{Theorem} 
\newtheorem{lemma}[theorem]{Lemma}

\newtheorem{definition}[theorem]{Definition}

\newcommand{\pgrams}{\ensuremath{\mathcal{G}}}
\newcommand{\eps}{\ensuremath{\varepsilon}}
\newcommand{\typea}{\textsf{A}\xspace}
\newcommand{\typeb}{\textsf{B}\xspace}
\newcommand{\typec}{\textsf{C}\xspace}
\newcommand{\typed}{\textsf{D}\xspace}
\newcommand{\typee}{\textsf{E}\xspace}
\newcommand{\typef}{\textsf{F}\xspace}
\newcommand{\typex}{\textsf{X}\xspace}
\newcommand{\stc}{\textsf{S}\xspace}
\newcommand{\lmr}{\textsf{LMR}\xspace}
\newcommand{\lmrs}{\textsf{LMR}s\xspace}
\newcommand{\dsc}{\textsf{DS}\xspace}
\newcommand{\dscs}{\textsf{DS}s\xspace}
\newcommand{\bc}{\textsf{BC}\xspace}
\newcommand{\bcs}{\textsf{BC}s\xspace}
\newcommand{\sidec}{\textsf{sc}\xspace}
\newcommand{\cornerc}{\textsf{cc}\xspace}

\newcommand{\Reals}{\ensuremath{\mathbb{R}}}
\newcommand{\rects}{\ensuremath{\mathcal{R}}}
\newcommand{\set}[1]{\left\lbrace #1\right\rbrace}
\newcommand{\vis}{\ensuremath{{V}}}

\newcommand{\blc}{\ensuremath{c_\theta}}
\newcommand{\rth}{\ensuremath{R_\theta}}
\newcommand{\gth}{\Gamma_\theta}
\newcommand{\feta}{\bar{\eta}}
\newcommand{\fdelta}{\bar{\delta}}
\newcommand{\flambda}{\bar{\lambda}}
\def\polylog{\operatorname{polylog}}
\newcommand{\ob}{\textsf{ob}\xspace}

\newbox\ProofSym \setbox\ProofSym=\hbox{%
	\unitlength=0.18ex%
	\begin{picture}(10,10) \put(0,0){\framebox(9,9){}}
	\put(0,3){\framebox(6,6){}}
	\end{picture}}

\newenvironment{denseitems}{\list{$\bullet$}{\itemsep2pt\parsep0pt}}{\endlist}

\title{Maximum-Area Rectangles in a Simple Polygon\thanks{This research was supported by the MSIT(Ministry of Science and ICT), Korea, under the SW Starlab support program(IITP--2017--0--00905) supervised by the IITP(Institute of Information \& communications Technology Planning \& Evaluation).}}
\author{Yujin Choi\thanks{Technische Universit\"at Berlin, Germany, Email: \texttt{yj5162@postech.ac.kr}} \and
Seungjun Lee\thanks{Pohang University of Science and Technology, Pohang, Korea, Email: \texttt{\{juny2400, heekap\}@postech.ac.kr}} \and
Hee-Kap Ahn\footnotemark[3]}

\date{}
\begin{document}
\maketitle

\begin{abstract}
  We study the problem of finding maximum-area rectangles contained
  in a polygon in the plane. There has been a fair amount of work for this problem 
  when the rectangles have to be axis-aligned or when the polygon is convex. We consider this problem
  in a simple polygon with $n$ vertices, possibly with holes, 
  and with no restriction on the orientation of
  the rectangles. We present an algorithm that computes a maximum-area rectangle
  in $O(n^3\log n)$ time using $O(kn^2)$ space, where $k$ is the number of reflex vertices
  of $P$. Our algorithm can report all maximum-area rectangles
  in the same time using $O(n^3)$ space.
  We also present a simple algorithm that finds a
maximum-area rectangle contained in a convex polygon with $n$ vertices
in $O(n^3)$ time using $O(n)$ space. 
\end{abstract}

\section{Introduction}
Computing a largest figure of a certain prescribed shape contained in
a container is a fundamental and important optimization problem in
pattern recognition, computer vision and computational geometry.
There has been a fair amount of work for finding rectangles of maximum
area contained in a convex polygon $P$ with $n$ vertices in the plane.
Amenta showed that an axis-aligned rectangle of maximum area can be
found in linear time by phrasing it as a convex programming
problem~\cite{Amenta-1994}.  Assuming that the vertices are given in
order along the boundary of $P$, stored in an array or balanced binary
search tree in memory, Fischer and H\"offgen gave $O(\log^2 n)$-time
algorithm for finding an axis-aligned rectangle of maximum area
contained in $P$~\cite{Fischer-1994}.  The running time was improved
to $O(\log n)$ by Alt et al.~\cite{Alt-1995}.

Knauer et al.~\cite{Knauer-2012} studied a variant of the problem in
which a maximum-area rectangle is not restricted to be axis-aligned
while it is contained in a convex polygon. They gave randomized and
deterministic approximation algorithms for the problem.
Recently, Cabello et al.~\cite{Cabello-2016} gave an exact
$O(n^3)$-time algorithm for finding a maximum-area rectangle with no
restriction on its orientation that is contained in a convex $n$-gon.
They also gave an algorithm for finding a maximum-perimeter rectangle
and approximation algorithms.

This problem has also been studied for containers which are not
necessarily convex.  Some previous work on the problem focuses on
finding an axis-aligned rectangle of maximum area or of maximum
perimeter contained in a rectilinear polygonal container in the
plane~\cite{Aggarwal-1988,McKenna-1985,Wood-1988}.  Daniels et
al. studied the problem of finding a maximum-area axis-aligned
rectangle contained in a polygon, not necessarily convex and possibly
having holes~\cite{Daniels-1997}.  They gave $O(n\log^2 n)$-time
algorithm for the problem. Later, Boland and Urrutia improved the
running time by a factor of $O(\log n)$ for simple polygons with $n$
vertices~\cite{Boland-2001}.  With no restriction on the orientation
of the rectangles, Hall-Holt et al.~gave a PTAS for finding a
fat\footnote{A rectangle is $c$-fat if its aspect ratio is
  at most $c$ for some constant $c$.} rectangle of maximum area
contained in a simple polygon~\cite{Hall-Holt2006}.

\subsection{Our results}
We study the problem of finding a maximum-area rectangle
with no restriction on its orientation that is contained in a simple
polygon $P$ with $n$ vertices, possibly with holes, in the plane.  We
are not aware of any previous work on this problem, except a PTAS for
finding a fat rectangle of maximum area inscribed in a simple
polygon~\cite{Hall-Holt2006}.  We present an algorithm that computes a
maximum-area rectangle contained in a simple polygon with $n$ vertices
in $O(n^3\log n)$ time using $O(kn^2)$ space, where $k$ is the number
of reflex vertices of $P$. Our algorithm can also find all rectangles
of maximum area contained in $P$ in the same time using $(n^3)$
space. We also present a simple algorithm that finds a maximum-area
rectangle contained in a convex polygon with $n$ vertices in $O(n^3)$
time using $O(n)$ space.

To obtain the running time and space complexities, we characterize the
maximum-area rectangles and classify them into six types, based on the
sets of contacts on their boundaries with the polygon boundary.  Then
we find a maximum-area rectangle in each type so as to find a
maximum-area rectangle contained in $P$.  To facilitate the process,
we construct a ray-shooting data structure for $P$ of $O(n)$ space
which supports, for a given query point in $P$ and a direction,
$O(\log n)$ query time.  We also construct the visibility region from
each vertex within $P$, which can be done in $O(n^2)$ time using
$O(n^2)$ space in total.  For some types, we compute locally maximal
rectangles aligned to the coordinate axes while we rotate the axes. To
do this, we maintain a few data structures such as double staircases
of reflex vertices and priority queues for events during
the rotation of the coordinate axes. They can be constructed and
maintained in $O(kn^2\log n)$ time using $O(kn^2)$ space.  The total
number of events considered by our algorithm is $O(n^3)$, each of
which is handled in $O(\log n)$ time.
\begin{theorem}
  \label{thm:final}
  We can compute a largest rectangle contained in a simple polygon
  with $n$ vertices, possibly with holes, in
  $O(n^3\log n)$ time using $O(kn^2)$ space, where $k$ is the number
  of reflex vertices. We can report all largest rectangles in the same
  time using $O(n^3)$ space.
\end{theorem}
\begin{theorem}
  \label{thm:convex}
  We can find a largest rectangle in a convex polygon $P$ with $n$
  vertices in $O(n^3)$ time using $O(n)$ space.
\end{theorem}

\section{Preliminary}
Let $P$ be a simple polygon with $n$ vertices in the plane.  For ease
of description, we assume that $P$ has no hole.  When $P$ has holes,
our algorithm works with a few additional data structures and
procedures for testing if candidate rectangles contain a hole. We
discuss this in Section~\ref{sec:with-holes}.  Without loss of
generality, we assume that the vertices of $P$ are given in order
along the boundary of $P$.  We denote by $k$ with $k \ge 1$ the number of
reflex vertices of $P$. We assume the general position condition that
no three vertices of $P$ are on a line. Whenever we say a
\emph{largest rectangle}, it refers to a maximum-area rectangle
contained in $P$.

We use the $xy$-Cartesian coordinate system and rotate the $xy$-axes
around the origin while the polygon is stationary. We use $\mathsf{C}_\theta$
to denote the coordinate axes obtained by rotating the $xy$-axes of the 
standard $xy$-Cartesian coordinate system by $\theta$ degree 
counterclockwise around the origin.  For a
point $p$ in the plane, we use $p_x$ and $p_y$ to denote the $x$- and
$y$-coordinates of $p$ with respect to the coordinate axes,
respectively. We say a segment or line is \textit{horizontal} (or
\textit{vertical}) if it is parallel to the $x$-axis (or the
$y$-axis).  Let $\eta(p)$, $\lambda(p)$ and $\delta(p)$ denote the ray
(segment) emanating from $p$ going horizontally leftwards, rightwards
and vertically downwards in the coordinate axes, respectively, until
it escapes $P$ for the first time.  We call the endpoint of a ray
other than its source point the \textit{foot} of the ray.  We denote
the foot of $\eta(p)$, $\lambda(p)$ and $\delta(p)$ by $\feta(p)$,
$\flambda(p)$ and $\fdelta(p)$, respectively.

We use $D_\eps(p)$ to denote the disk centered at a point $p$ with
radius $\eps>0$.  For any two points $p$ and $q$ in the plane, we use
$pq$ and to denote the line segment connecting $p$ and $q$, and $|pq|$
to denote the length of $pq$.  For a segment $s$, we use $D(s)$ to
denote the smallest disk containing $s$.  For a subset $S\subseteq P$,
we define the \emph{visibility region} of $S$ as
$\vis(S)=\{x\in P\mid px\subset P$ for every points $p \in S\}$.  For
a point $p\in P$, we abuse the notation such that
$\vis(p)=\vis(\{p\})$. For a set $X$, we use $\partial X$ to denote
the boundary of $X$.

\subsection{Existence of a maximum-area rectangle in a simple polygon}
We start with showing the existence of a maximum-area
rectangle contained in a simple polygon.
The set $\pgrams$ of all
parallelograms in the plane is a metric space under the Hausdorff
distance measure $d_H$. 
The Hausdorff distance between two
sets $A$ and $B$ of points in the plane is defined as
$d_H(A,B)= \max\{\sup_{a\in A} \inf_{b\in B} d(a,b), \sup_{b\in B}
\inf_{a\in A} d(a,b)\}$, where $d(a,b)$ denotes the distance between
$a$ and $b$ of the underlying metric. Since the area function
$\mu: \pgrams \rightarrow \Reals^{\ge 0}$ is continuous in 
$\pgrams$,
the following lemma assures the existence of a largest
rectangle contained in $P$ and thus justifies the problem.  Let
$\rects$ denote the set of all rectangles contained in $P$. Clearly,
$\rects\subset\pgrams$.  
\begin{lemma}
  The set $\rects$ is compact.
\end{lemma}
\label{lem:rectangles.compact}
\begin{proof}
  Define $f: \Reals^6 \rightarrow \pgrams$ to be a function that maps
  a triplet $(p,u,v)$ of points $p$, $u$, and $v$ in $\Reals^2$ to the
  parallelogram in $\Reals^2$ that has $p$, $p+u$, $p+v$, and $p+u+v$
  as the four corners.

  If a parallelogram $G\in \pgrams$ is not contained in $P$, there
  always exists a point $q\in G$ and a disk $D_\eps(q)$ for some
  $\eps>0$ satisfying $D_\eps(q)\cap P=\emptyset$ in the plane.  Then
  for any parallelogram $Q\in\pgrams$ with $d_H(G,Q) < \eps$, the
  intersection $Q\cap D_\eps(q)$ is not empty and $Q$ is not contained
  in $P$. Thus, $\mathcal{C}=\set{G\in \pgrams\mid G\not\subset P}$ is
  open in $\pgrams$, and therefore
  $W_P = \set{ (p,u,v) \in \Reals^6\mid f(p,u,v) \subset
    P}=f^{-1}(\pgrams\setminus \mathcal{C})$ is closed.  This implies
  that
  $T_P = \set{ (p,u,v) \in \Reals^6\mid f(p,u,v) \subset P, \text{a
      rectangle}} = W_P \cap \set{(p,u,v)\mid u\cdot v = 0}$ is closed
  and also bounded in $\Reals^6$, i.e. compact.  Now we can conclude
  that $f(T_P)=\rects$ is also compact by $f$ being continuous in
  $\Reals^6$.
\end{proof}
\subsection{Classification of largest rectangles}
We give a classification of largest rectangles based on the sets of
contacts they have on their boundaries with the polygon boundary.  We
say a rectangle contained in $P$ has a \emph{side-contact}
(\sidec for short) if a side has a reflex vertex of $P$ lying on
it, excluding the corners.  Similarly, we say a rectangle contained in
$P$ has a \emph{corner-contact} (\cornerc for short) if a corner
lies on an edge or a vertex of $P$. 
For instance, the rectangle of case $\typed_1$ in Figure~\ref{fig:DSC-classification}
has a side-contact on its top side and corner-contacts on its corners,
except its bottom-right corner.
When two opposite corners (or two opposite sides) have corner-contacts
(or side-contacts), we say the contacts are \emph{opposite}.

Daniels et al.~\cite{Daniels-1997} studied this problem with
restriction that rectangles must be axis-aligned.  
They presented a classification of
determining sets of contacts, defined below, into five types for a
largest axis-aligned rectangle contained in a simple polygon in the
plane.

\begin{definition}[Determining set of contacts~\cite{Daniels-1997}]
  \label{def:dsc}
  A set $Z$ of contacts is \textit{a determining set of contacts} if
  the largest axis-aligned rectangle satisfying $Z$ has finite area
  and the largest axis-aligned rectangle satisfying any proper subset
  of $Z$ has greater or infinite area.
\end{definition}

In our problem, a largest rectangle $R$ is not necessarily
axis-aligned.  Consider two orthogonal lines which are parallel to the
sides of $R$ and pass through the origin.  Since $R$ is aligned to the
coordinate axes defined by the lines, it also has a determining set of
contacts defined by Daniels et al.  From this observation, we present
a classification of the determining sets of contacts (\dsc for short)
for a largest rectangle in $P$ into six \emph{canonical types}, from
\typea to \typef, and their subtypes. The classification is given
below together with figures in Figure~\ref{fig:DSC-classification}.
\begin{figure}[t]
  \centering \includegraphics[width=.9\textwidth]{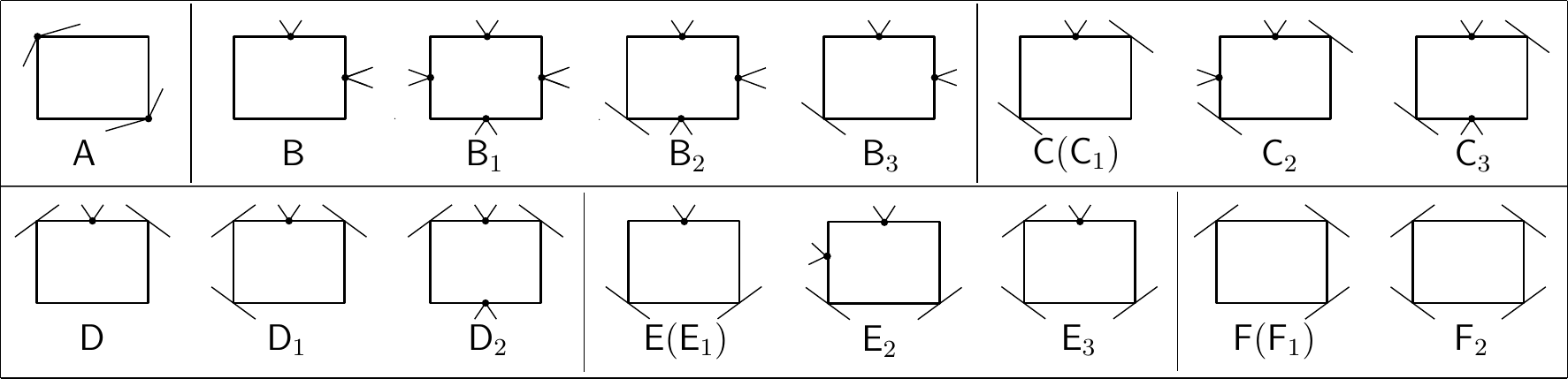}
  \caption{Classification of the determining sets of contacts of
    largest rectangles when rotations are allowed. Each canonical type
    \textsf{X}, except \typea, has a few subtypes $\typex_i$ for
    $i=1,2,3$.}
  \label{fig:DSC-classification}
\end{figure}

\begin{denseitems}
\item[-]Type \typea.  Exactly two opposite \cornerc{s} lying on
  convex vertices of $P$.
\item[-]Type \typeb.  One \sidec on each side incident to a
  corner $c$.  In addition, $\typeb_1$ has a \sidec on each of
  the other two sides, and $\typeb_2$ and $\typeb_3$ have a
  \cornerc on the corner $c'$ opposite to $c$. $\typeb_2$
  has another \sidec on a side incident to $c'$.
\item[-]Type \typec ($\typec_1$).  Two \cornerc{s} on opposite
  corners $c$ and $c'$, and a \sidec on a side $e$ incident to
  a corner $c$.  $\typec_2$ has another \sidec on the side
  incident to $c'$ and adjacent to $e$, and $\typec_3$ has another
  \sidec on a side opposite to $e$.
\item[-]Type \typed.  A \sidec on a side $e$ and a
  \cornerc on each endpoint of $e$. $\typed_1$ has another
  \cornerc and $\typed_2$ has another \sidec on the side
  opposite to $e$.
\item[-]Type \typee ($\typee_1$). A \sidec on a side $e$ and
  a \cornerc on each endpoint of the side $e'$ opposite to
  $e$.  $\typee_2$ has another \sidec on a side other than $e$
  and $e'$.  $\typee_3$ has another \cornerc on an endpoint of
  $e$.
\item[-]Type \typef ($\typef_1$).  \cornerc{s} on three corners.
  $\typef_2$ has \cornerc{s} on all four corners.
\end{denseitems}

This classification is the same as the one by Daniels et al., except
for types \typea, \typee, and \typef. We subdivide the last type in the
classification by Daniels et al. into two types, \typee
and \typef, for ease of description.  A \dsc for
type \typea consists of exactly two opposite corner-contacts lying on
convex vertices of $P$ while the corresponding one by Daniels et
al.~\cite{Daniels-1997} has two opposite corners lying on the boundary
(not necessarily on vertices) of $P$. This is because there is no
restriction on the orientation of the rectangle, which is
shown in Lemma~\ref{lem:type-A-convex}.
\begin{lemma}\label{lem:type-A-convex}
  If a largest rectangle has no contact on its boundary, except two
  corner-contacts at opposite corners $c$ and $c'$, then for every
  edge $e$ of $P$ incident to $c$ or $c'$, $D_\eps(c)\cap e$ or
  $D_\eps(c')\cap e$ is contained in $D(cc')$ for some $\eps>0$,
  respectively. Thus, the corner-contacts are on convex vertices of
  $P$.
\end{lemma}
\begin{proof}
  Let $R$ be a largest rectangle that has no contact on its boundary,
  except two corner-contacts at opposite corners $c$ and $c'$.  Assume
  to the contrary that there is no $\eps>0$ satisfying
  $D_\eps(c)\cap e\subset D(cc')$ for an edge $e$ incident to
  $c$. Then we can get another rectangle $R'$ by rotating $R$ around
  $c'$ in a direction such that $R'$ has no contact on its boundary,
  except the corner-contact at $c'$. Then we can get a larger
  rectangle by expanding $R'$ from $c'$ which contradicts the
  optimality of $R$ under the area function $\mu$.  We can prove the
  claim for edges incident to $c'$ analogously.
\end{proof}

\subsection{Maximal and breaking configurations.}
Recall that $\rects$ denotes the set of all rectangles contained in
$P$.  Our algorithm finds a largest rectangle in $\rects$ of each
(sub)type so as to find a largest rectangle contained in $P$.  We call
a rectangle that gives a \emph{local maximum} of the area function
$\mu$ among rectangles in $\rects$ a \emph{local maximum rectangle}
(\lmr for short).  We say an \lmr $R$ is of type $\typex$ if $R$ has
all contacts of subtype $\typex_i$ for some $i=1,2,\ldots$.  Since a
largest rectangle contained in $P$ is a rectangle aligned to the axes
that are parallel to its sides, it has contacts of at least one type
defined above and is an \lmr of that type.  Therefore, our algorithm
finds a largest rectangle among all possible \lmrs of each type.
If there are multiple largest rectangles, our algorithm finds
them all but keeps one of them.

Consider a rectangle $R\in\rects$ that satisfies a \dsc $Z$.  If there
is no contact other than $Z$, there exists a continuous transformation
of $R$ such that the transformed rectangle is a rectangle contained in
$P$ and satisfying $Z$. Then by such a continuous 
transformation the area of $R$ may
change. 
Imagine we continue with such a transformation
until the transformed rectangle $R'$ 
gets another contact. In this case, we say $R'$ is in a
\textit{breaking configuration} (\bc for short) of $Z$.  During the
transformation, the area of $R'$ may become locally maximum.  If $R'$
is locally maximum and has no contact other than $Z$,
we say $R'$ is in a \textit{maximal configuration} of $Z$.  There can
be $O(1)$ maximal configurations of $Z$, which can be observed from
the area function of the rectangle.
    
\begin{lemma}\label{lem:configurations} An \lmr
  satisfying $Z$ is in a maximal configuration or a breaking
  configuration.\end{lemma}

For a breaking configuration $Z'$ of a \dsc $Z$, observe that
$Z'\setminus Z$ is a singleton and $Z'$ can be a \bc of some other
\dsc{s}. With this fact, we can classify \bcs by avoiding repetition
and reducing them up to symmetry. See Figure~\ref{fig:BCs} for
breaking configurations.

We use $\gth(Z)$ to denote the axis-aligned rectangle of
  largest area that satisfies a \dsc $Z$ in $\mathsf{C}_\theta$.
We say a \dsc $Z$ is
  \textit{feasible} at an orientation $\theta$ if $\gth(Z)$ is a
  rectangle contained in $P$.  We say an orientation $\theta$ is
  \textit{feasible} for $Z$ if $\gth(Z)$ is contained in $P$.

\section{Computing a largest rectangle of type \typea}
\label{sec:comp-A}
We show how to compute all \lmrs of type \typea and a largest
rectangle among them. 
\begin{figure}[t] \centering
  \includegraphics[width=.55\textwidth]{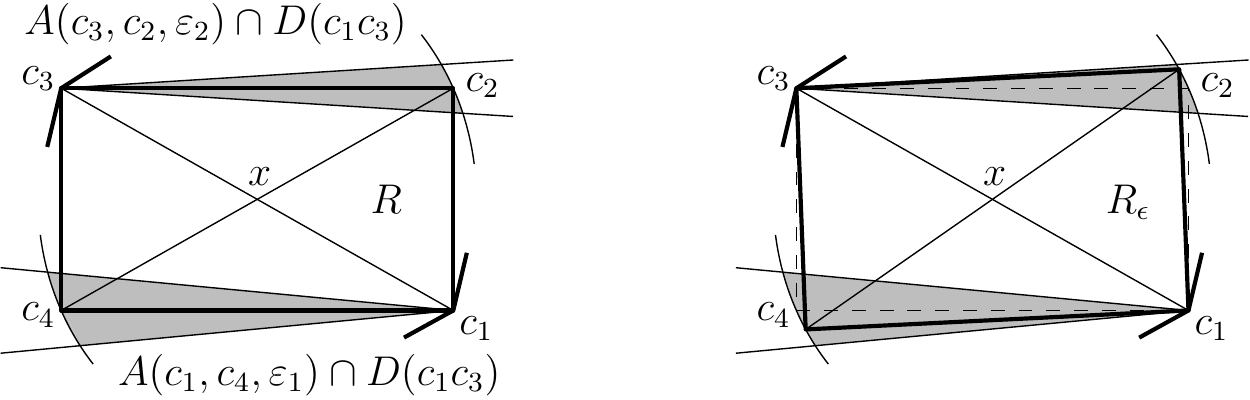}
  \caption{If an \lmr $R$ of type \typea is not a square, there is a
    larger rectangle $R_\eps$.}
  \label{fig:type-A-square}
\end{figure}
\begin{lemma}\label{lem:type-A-square}
  Every \lmr of type \typea is a square.
\end{lemma}
\begin{proof}
  Suppose there is an \lmr $R$ of type \typea which is not a square.
  Without loss of generality, assume that $R$ has corners
  $c_1, c_2, c_3$, and $c_4$, from its bottom-right corner $c_1$, in
  counterclockwise order such that $R$ has corner-contacts at $c_1$
  and $c_3$ on convex vertices of $P$ and
  $\angle c_2 x c_1 < \frac{\pi}{2} < \angle c_3 x c_2$ for
  $x = c_1 c_3 \cap c_2 c_4$.
  See Figure~\ref{fig:type-A-square}.  Since both $c_2$ and $c_4$ are
  contained in the interior of $P$, there exist
  $\eps_1, \eps_2 \in (0, \frac{\pi}{4}-\frac{1}{2} \angle c_2 x c_1)$
  such that $A(c_1, c_4, \eps_1) \cap D(c_1c_3) \subset P$ and
  $A(c_3, c_2, \eps_2) \cap D(c_1c_3) \subset P$, where
  $A(p, q, \theta) = \{ x\in \mathbb{R}^2 \mid \angle xpq \in (0,
  \theta)\}$ for $\theta>0$.
  
  Let $R_\eps$ denote the rectangle with diagonals $c_1 c_3$ and the
  segment obtained by rotating $c_2 c_4$ by
  $\eps\in (0, 2\min\{\eps_1 , \eps_2\})$ around $x$ in
  counterclockwise direction.  Then $R_\eps$ is contained in the union
  of $R$, $A(c_3, c_2, \eps_2) \cap D(c_1c_3)$, and
  $A(c_3, c_2, \eps_2) \cap D(c_1c_3)$, which implies
  $R_\eps \subset P.$ Moreover, $\mu(R_\eps )>\mu(R)$ for all $\eps$,
  which contradicts to the assumption that $R$ is an \lmr.
\end{proof}
By Lemma~\ref{lem:type-A-square}, it
suffices to check all possible squares in $P$ with two opposite
corners on convex vertices of $P$.  Since $\partial P$ is a simple
closed curve, we can determine if a rectangle $R$ is contained in $P$
by checking if all four sides of $R$ are contained in $P$.

\begin{lemma}\label{lem:lmrs-type-A}
  We can compute a largest rectangle among all \lmrs of type \typea in
  $O((n-k)n+(n-k)^2 \log n)$ time using $O(n)$ space.
\end{lemma}
\begin{proof}
  We first construct a ray-shooting data structure, such as the one by
  Hershberger and Suri~\cite{Hershberger-1995}, for $P$ in $O(n)$ time
  using $O(n)$ space, which supports a ray-shooting query in
  $O(\log n)$ time.  Then for each convex vertex $v$ of $P$, we
  compute the list $L$ of the vertices of $P$ that are visible from
  $v$ using $\vis(v)$ in $O(n)$ time using $O(n)$ space.
  
  Then for each convex vertex $v'$ in $L$, we determine if the square
  $S$ with diagonal $vv'$ is contained in $P$ by checking if each side
  of $S$ is contained in $P$ with a ray-shooting query from an
  endpoint along the side in $O(\log n)$ time.  The length of the ray
  from the query is at least $|vv'|/\sqrt{2}$ if and only if the side
  is contained in $P$.  If all four sides of $S$ are contained in $P$,
  we compute the area of $S$. Among all such squares, we return one
  with the maximum area.
\end{proof}

\section{Computing a largest rectangle of type \typeb}
\label{sec:comp-b}
We show how to compute all \lmrs of type \typeb and a largest
rectangle among them.  We compute for each \dsc a largest \lmr over
the maximal and breaking configurations. In doing so, we maintain a
combinatorial structure for each reflex vertex which helps compute all
\lmr{s} of type \typeb during the rotation of the coordinate axes.

\subsection{Staircase of a point in a simple polygon}
We define the staircase $\stc(u)$ of a point
$u\in P$ as the set of points $p\in P$ with
$p_x \le u_x$ and $p_y \le u_y$ such that the axis-aligned rectangle
with diagonal $up$ is contained in $P$ but no axis-aligned rectangle
with diagonal $uq$ is contained in $P$ for any point $q\in P$
with $q_x < p_x$ and $q_y < p_y$.  Thus, $\stc(u)$ can be
represented as a chain of segments. See
Figure~\ref{fig:extreme-staircase} (a).

The staircase of a point $u\in P$ in orientation $\theta$, 
denoted by $\stc_\theta(u)$, is defined as the staircase of $u$ in
$\mathsf{C}_\theta$.
Every
axis-aligned segment of $\stc_\theta(u)$ has one endpoint at a vertex
of $P$, $\feta(u)$, or $\fdelta(u)$.  A segment of
$\stc_\theta(u)$ that is not aligned to the axes is 
a part of an edge $e$ of $P$ and is called an \emph{oblique segment}.
(We say $e$ appears to the staircase in this case.)
Each vertex of $\stc_\theta(u)$ which is a polygon vertex or a foot of
$\eta(u)$ or $\delta(u)$ is called an \emph{extremal vertex}.  An
extremal vertex $v$ is called a \emph{tip} if it is a reflex vertex of
$P$.  A vertex of $\stc_\theta(u)$ contained in $P$ is called a
\emph{hinge}.
A \emph{step} of $\stc_\theta(u)$, denoted by an ordered pair $(a,b)$,
is the part of $\stc_\theta(u)$ between two consecutive extremal
vertices $a$ and $b$ along $\stc_\theta(u)$, where $a_x\leq b_x$ and
$a_y\geq b_y$.  It consists of either (A) two consecutive segments
$ar$ and $rb$ for $r=\delta (a) \cap \eta(b)$, or (B)
three consecutive segments $a\fdelta(a)$, $\fdelta(a)\feta(b)$, and $\feta(b)b$
with $\fdelta(a)_x\leq \feta(b)_x$ and $\fdelta(a)_y\geq \feta(b)_y$. 
Note that $\fdelta(a)\feta(b)$ is the oblique segment of step $(a,b)$ 
which we denote by $\ob(a,b)$.  A horizontal,
vertical, or oblique segment of a step can be just a point in case of
degeneracy.

We can construct $\stc_\theta(u)$ for a fixed $\theta$ in
$O(n)$ time by traversing the boundary of $P$ in counterclockwise direction
starting from $\feta(u)$ while maintaining the staircase of $u$ with
respect to the boundary chain traversed so far. When the next vertex
$v$ of the boundary chain satisfies $v_x\geq t_x$ and $v_y\leq t_y$ 
for the last vertex $t$ of the current staircase,
we append it to the staircase. If (part of) the edge incident to $v$ is an
oblique segment of the staircase, then we append it together with $v$
to the staircase.  If $v_x<t_x$,  
we ignore $v$ and proceed to the vertex next to $v$. 
If $v_x\geq t_x$ and $v_y>t_y$,
we remove the portion of the current staircase violated by $v$ and
append $v$ to the staircase accordingly. Observe that each vertex 
and each edge appear on the staircase at most once during the
construction. 

\subsection{Maintaining the staircase during rotation of the
  coordinate system}
\label{sec:type-B-staircase}
Bae et al.~\cite{Bae-2009} considered the rectilinear convex hull for
a set $Q$ of $n$ point in the plane and presented a method of
maintaining it while rotating the coordinate system in $O(n^2)$
time. The boundary of the rectilinear convex hull consists of four
maximal chains, each of which is monotone to the coordinate axes.  
We adopt their method and maintain the staircase of a reflex vertex 
$u$ in a simple polygon.

The combinatorial structure of $\stc_\theta(u)$ changes during the
rotation. Figure~\ref{fig:extreme-staircase} (b-f) show
$\stc_{0}(u)$, $\stc_{\theta_1}(u)$ and $\stc_{\theta_2}(u)$ for
three orientations $0,\theta_1,$ and $\theta_2$ ($0<\theta_1<\theta_2<\pi/2$).  Two consecutive steps,
$(a,b)$ and $(b,c)$, of the staircase merge into one step $(a,c)$ when
$\fdelta(a)$ meets $b$ (Figure~\ref{fig:extreme-staircase} (b)).  A
step $(a,b)$ splits up into two steps $(a,v)$ and $(v,b)$ when
$\feta(b)$ meets a polygon vertex $v$
(Figure~\ref{fig:extreme-staircase} (c)).  A step changes its type
between (A) and (B) when the hinge of a step hits a polygon edge (and
then it is replaced by an oblique segment) or the oblique segment of a
step degenerates to a point (and then it becomes a hinge)
(Figure~\ref{fig:extreme-staircase} (d)).  The upper tip $a$ (or the
lower tip $b$) of a step $(a,b)$ can disappear from the staircase when
$\fdelta(\feta(u))$ meets $a$ (or $\feta(\fdelta(u))$ meets $b$).
Finally, a vertex, possibly along with an edge incident to it, can
be added to or deleted from $\stc_\theta(u)$ when it is met by 
$\feta(\fdelta(u))$ or $\fdelta(\feta(u))$. We call such a change of the staircase due to the
cases described above a \emph{step event}.

One difference of the staircase $\stc_\theta(u)$ to the one for a
point-set by Bae et al.  is that the two boundary points of
$\stc_\theta(u)$ are $\feta(u)$ and $\fdelta(u)$. Since the
polygon is not necessarily monotone with respect to the axes, the
staircase may change discontinuously when $\feta(u)$ or $\fdelta(u)$
meets a vertex of $P$, which we call a \emph{ray event}.  The step of
$\stc_\theta(u)$ incident to $\feta(u)$ is replaced by a chain of
$O(n)$ steps when $\feta(u)$ meets a vertex of $P$
(Figure~\ref{fig:extreme-staircase} (e)).  A subchain 
incident to $\fdelta(u)$ is replaced by a single step when $\fdelta(u)$ 
meets a vertex of $P$ (Figure~\ref{fig:extreme-staircase} (f)).  
We call the appearance or disappearance of a step caused by 
a ray event a \textit{shift event} of the ray event.  
Note that $O(n)$ shift events occur at a ray event.
Observe that all the changes occurring in a staircase during the
rotation are caused by step, ray, or shift events. We abuse 
$\stc_\theta(u)$ to denote the combinatorial structure
of the staircase if understood in context.
\begin{figure}[t]
  \centering
  \includegraphics[width=\textwidth]{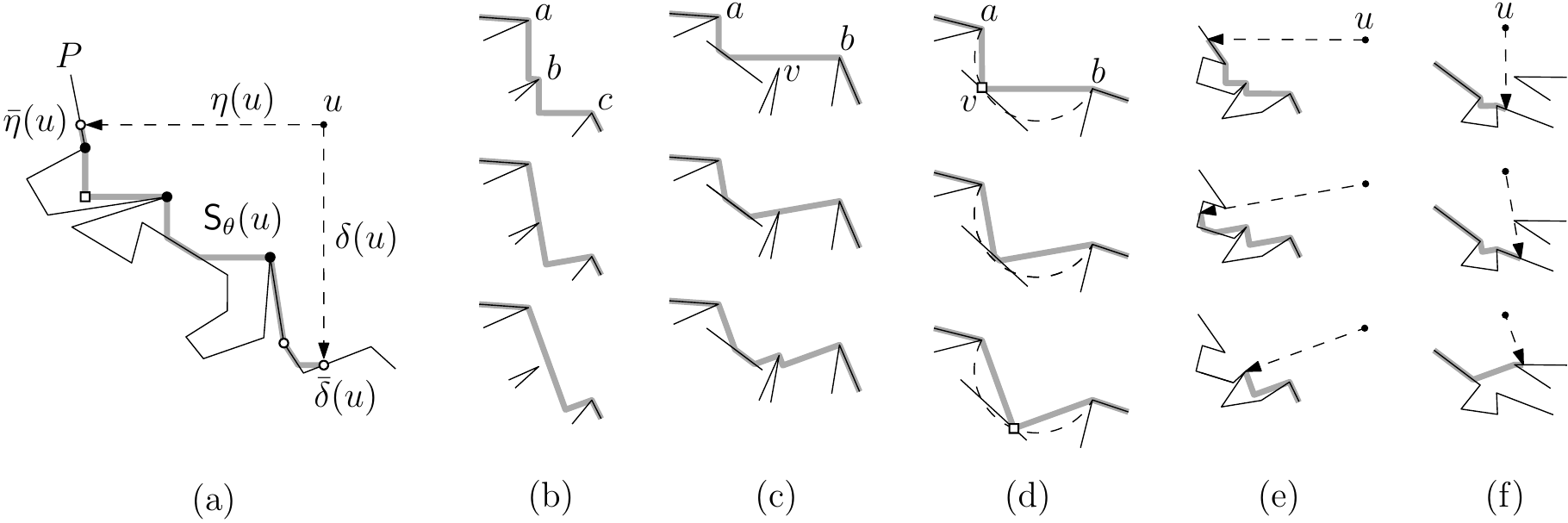}
  \caption{(a) Staircase $S_\theta(u)$ (thick gray chain) and the tips
    (black disks), the extremal vertices (black disks and circles),
    and the hinge (square) of $S_\theta(u)$. (b--d) Step events, and
    (e--f) ray events during the rotation of the coordinate system.}
  \label{fig:extreme-staircase}
\end{figure}

\begin{lemma}\label{lem:hinge_edge}
  Suppose a polygon edge $e$ disappears from a step $(a,b)$
  of $\stc_\theta(u)$ at $\theta$ by a step event.  Then
  $e$ never appears again to the staircase during the remaining
  rotation until $a_y$ becomes larger than $u_y$ in the coordinate
  system.
\end{lemma}
\begin{proof}
  Observe that a polygon edge $e$ of $P$ appearing on
  $\stc_\theta(u)$ is the oblique segment $\ob(a,b)$ of a step
  $(a, b)$ which is contained in $D(ab)$.
  We claim that $e$ disappears from the staircase at
  $\theta$ when $\ob(a,b)$ degenerates to a point on $e$
  and becomes the hinge of $(a,b)$.  Assume to the contrary that
  $e$ disappears from the staircase when $\feta(b)$ meets a polygon vertex $v$. 
  Then $v$ becomes a new extremal vertex of the
  staircase by definition, and $e$ remains on
  $\stc_\theta(u)$ as an oblique segment $\ob(a,v)$ of step $(a,v)$.
  Similarly, it can be shown for the other types of step events that
  the $e$ remains on $\stc_\theta(u)$ as the oblique segment
  of a step in the staircase right after such an event.

  Now we prove that $e$ never appears again to the staircase
  as the oblique segment of a step until $a_y$ becomes larger than
  $u_y$ in the coordinate system.  Let $\theta_1$ be the orientation
  with $\theta<\theta_1<\pi/2$ that aligns $a$ and $b$ vertically in
  the coordinate system rotated by $\theta_1$.  Note that for any
  $\theta' \in (\theta,\theta_1]$, every point $p\in e$ satisfies
  $p_x< a_x$ or $p_y< b_y$, which implies that the axis-aligned
  rectangle with diagonal $up$ contains $a$ or $b$ in its interior.
  Since $e$ is an edge on the boundary chain from $a$ to $b$
  counterclockwise along the boundary of $P$, it never appears to the
  staircase during the remaining rotation after $\theta_1$ until $a_y$
  becomes larger than $u_y$.
\end{proof}
\begin{lemma}\label{lem:stair-complex}
  The number of events that occur to $\stc_\theta(u)$ during the
  rotation is $O(n^2)$.
\end{lemma}
\begin{proof}
  In the initialization at $\theta=0$, there are $O(n)$ vertices and
  edges in $\stc_0(u)$.  A ray event occurs only if $\feta(u)$ or
  $\fdelta(u)$ meets a vertex of $\vis(u)$.  Since $\feta(u)$ (and
  $\fdelta(u)$) meets a vertex of $P$ at most once during the
  rotation, the number of ray events is $O(n)$, resulting with
  $O(n^2)$ shift events.

  Now we count the step events occurring during the rotation. Once two
  consecutive steps, $(a,b)$ and $(b,c)$, merge into $(a,c)$ and tip
  $b$ disappears from $\stc_\theta(u)$ (by the event that $\fdelta(a)$
  meets $b$, Figure~\ref{fig:extreme-staircase} (b)), $b$ may appear
  again to the staircase only after $\feta(u)$ meets $a$.  Similarly,
  a step $(a,b)$ splits up into two steps, $(a,v)$ and $(v,b)$, and a
  polygon vertex $v$ appears as a tip to the staircase when $\feta(b)$
  meets $v$ (Figure~\ref{fig:extreme-staircase} (c)).  Since
  $\feta(v)$ and $\fdelta(v)$, for any vertex $v$, meet another vertex
  of $P$ at most once each during the rotation,
  a vertex appears to and disappears from the staircase $O(n)$ times
  in this way.  Moreover, the number of the step events occurring when
  $\fdelta(\feta(u))$ or $\feta(\fdelta(u))$ meets a polygon vertex
  $v$ is $O(n)$ for each vertex $v$.  Thus the total number of step
  events occurred by $\fdelta(\feta(u))$ and $\feta(\fdelta(u))$
  meeting polygon vertices
  is $O(n^2)$. It suffices to show that the total number of events
  induced by appearances or disappearances of oblique segments is
  $O(n^2)$. Lemma~\ref{lem:hinge_edge} implies that a
  polygon edge appears to the staircase as an oblique segment at most
  $O(n)$ times during the rotation. Thus the number of events induced
  by oblique segments is $O(n^2)$ in total.
\end{proof}
To capture these combinatorial changes and maintain the staircase
during the rotation, we construct for every reflex vertex $p$ of
$P$, 
the list of segments of visibility region $\vis(p)$ sorted in angular
order.  We compute for every pair $(p,q)$ of reflex vertices of $P$,
the list $C(p, q)$ of vertices and segments of
$\partial\vis(\{p,q\})\cap D(pq)$, 
sorted in angular order with respect to $p$
and $q$. We also compute for every pair $(p,e)$ of a reflex vertex $p$
and edge $e$, the sorted list $L(p,e)$ of angles at which
$\delta(\feta(p))$ or $\eta(\fdelta(p))$ meets a vertex of $P$ while
$\feta(p)$ or $\fdelta(p)$ lies on $e$.  We store for each orientation in $L(p,e)$
the information on the vertex corresponding to the orientation.  
This can be computed by
finding the points that $e$ intersects with the boundary of $D(tp)$
for each vertex $t$ of $P$.  These structures together constitute the
\emph{event map}.

We also construct an \emph{event queue} for each reflex vertex, which is a
priority queue that stores events indexed by their orientations.  This
is to update the staircase during the rotation in a way similar to the
one by Bae et al.~\cite{Bae-2009} using the event map.
\begin{lemma}\label{lem:stair-eventmap}
  The event map is of size $O(kn^2)$ and can be constructed in $O(kn^2\log n)$ time. 
\end{lemma}
\begin{proof}
  We construct event map by computing the visibility region for each
  vertex, $C(p,q)$ for every pair of reflex vertices, and
  $L(p,e)$ for every pair of a reflex vertex $p$ and an edge
  $e$ of $P$. $\vis(p)$ has $O(n)$ size and can be computed, 
  for each vertex $p$, in $O(n)$ time.  $C(p,q)$ has $O(n)$ size and can be computed in
  $O(n)$ time as well, for each reflex vertex $p$
  and each vertex $q$, by cutting $\vis(p)\cap\vis(q)$ 
  with $D(pq)$.
  Finally, when constructing $C(p,t)$, we mark
  $\partial D(tp) \cup e$ aligned on $e$ for vertex $t$ and edge $e$
  in $O(n)$ time for each $t$. Then we can sort the marks to get
  $L(p,e)$ in $O(n \log n)$ time using $O(n)$ space for each edge
  $e$. Therefore, it takes $O(kn^2 \log n)$ time to construct the event map
  and its size is $O(kn^2)$.
\end{proof}
For a reflex vertex $u$, we maintain $\stc_\theta(u)$ and the event
queue $\mathcal{Q}$ for $u$ during the rotation using the event map.
We also store the extremal vertices and edges of the staircase in a
balanced binary search tree $\mathcal{T}$ representing
$\stc_\theta(u)$ at the moment in order along the staircase 
so as to insert and delete an element in $O(\log n)$ time.  
We process
the events in the queue one by one in the order of priority and update
the event queue.
\begin{lemma}\label{lem:stair-update}
  Once the event map is constructed, $\stc_\theta(u)$ can be
  maintained over all events during the rotation in $O(n^2\log n)$
  time using $O(n)$ space.  The event queue for $u$ is maintained in
  the same time using $O(n^2)$ space.
\end{lemma}
\begin{proof}
  There are $O(n^2)$ events including $O(n)$ ray events in total by
  Lemma~\ref{lem:stair-complex}. Therefore, the event queue can be
  maintained in $O(n^2 \log n)$ time using $O(n^2)$ space.  Initially,
  we construct $\stc_0(u)$ and $\mathcal{T}$ in $O(n\log n)$ time.
  Consider a step event that occurs on step $(a,b)$ of
  $\stc_\theta(u)$: 
  $(a,b)$ and $(b,c)$ merge into $(a,c)$, $(a,b)$ splits
  up into two steps $(a,v)$ and $(v,b)$, $\ob(a,b)$ starts to appear
  in $(a,b)$, or $\ob(a,b)$ disappears from $(a,b)$.  We update
  $\stc_\theta(u)$ and $\mathcal{T}$ accordingly in $O(\log n)$ time.
  Then we apply binary search on $C(a,b)$ with query $\delta(a)$ or
  $\eta(b)$ to find the next candidate step event that may occur and
  change $(a,b)$ at an orientation $\theta'>\theta$ and to insert it
  into $\mathcal{Q}$. If the step contains $\feta(u)$ (or $\fdelta(u)$),
  we apply binary search on $L(u,e)$ and $\vis(v)$ with query
  $\delta(\feta(u))$ and $\eta(v)$ (or $\eta(\fdelta(u))$ and
  $\delta(v)$) to find next candidate step event and insert it into
  $\mathcal{Q}$.  This can also be done in $O(\log n)$ time.

  Now consider a ray event caused by $\eta(u)$ at which the tip $q=\feta(u)$ is
  removed from $\stc_\theta(u)$. We first construct $\stc_\theta(q)$
  and find the step on $\stc_\theta(q)$ that new $\feta(p)$ belongs
  to, where $p$ is the second tip of $\stc_\theta(u)$. From that step,
  we traverse $\stc_\theta(q)$ up to the first extremal vertex and add
  the steps encountered during the traverse to $\stc_\theta(u)$ (and
  to $\mathcal{T}$). For each new step, we insert the step event
  candidates corresponding to the step into $\mathcal{Q}$, after
  deleting the events associated with the steps disappearing from
  $\mathcal{Q}$.

  Similarly, at a ray event caused by $\delta(u)$ 
  when a vertex $q=\fdelta(u)$ is added to $\stc_\theta(u)$ as a tip, 
  we find the step of $\stc_\theta(u)$
  from $\mathcal{T}$ which contains $\feta(q)$, and delete the steps
  lying below $\feta(q)$, together with the step event candidates
  associated with those deleted steps from $\mathcal{Q}$. We also add
  the new step incident to $q$ with its step event to
  $\mathcal{Q}$.  There are $O(n)$ steps to delete. Since each
  deletion on $\mathcal{T}$ can be done in $O(\log n)$ time, a ray
  event and associated shift events can be handled in $O(n\log n)$
  time.
\end{proof}
By combining Lemmas~\ref{lem:stair-eventmap} and
\ref{lem:stair-update}, 
we have the following lemma.
\begin{lemma}
\label{lem:stair-complexity}
  The staircases of all $k$ reflex vertices of $P$ can be constructed
  and maintained in $O(kn^2\log n)$ time using $O(kn^2)$ space during
  the rotation.
\end{lemma}
\subsection{Data structures - double staircases, event map, and event
  queue}
Our algorithm computes all \lmrs of type \typeb during the rotation and
returns an \lmr with largest area among them.  To
do so, it maintains for each reflex vertex $u$ two staircases,
$S_\theta(u)$ and $S_{\theta+\frac{\pi}{2}}(u)$ which we call the
\emph{double staircase} of $u$, during the rotation of the coordinate
axes and computes the \lmrs of type \typeb
that have $u$ as the top \sidec.  Let $I$ denote the
interval of orientations such that the horizontal line with respect to
any $\theta\in I$ passing through $u$ is tangent to the
boundary of $P$ locally at $u$.  
Let $\rth$ be the largest
axis-aligned rectangle of type \typeb in $\theta\in I$
that is contained in $P$ and has $u$ as the top \sidec.
Observe that every reflex vertex lying on the right side is a tip of
$\stc_{\theta+\frac{\pi}{2}}(u)$. We use $X$ to denote the contact set
around the bottom-left corner $\blc$ of $\rth$.  Then $X$ contains (1)
a tip of $\stc_\theta(u)$ touching the left side and a tip on either
$\stc_\theta(u)$ or $\stc_{\theta+\frac{\pi}{2}}(u)$ touching the
bottom side (type $\typeb_1$), (2) an oblique segment $e$ on
$\stc_\theta(u)$ touching $\blc$ and a tip on either $\stc_\theta(u)$
or $\stc_{\theta+\frac{\pi}{2}}(u)$ touching the bottom side (type
$\typeb_2$), or (3) just an oblique segment $e$ on $\stc_\theta(u)$
touching $\blc$ (type $\typeb_3$). 

For a reflex vertex $u$ of $P$, we construct the double staircase of
$u$, $\stc_0(u)$ and $\stc_\frac{\pi}{2}(u)$. Then we maintain the
event queue $\mathcal{Q}$ containing event orientations in order: the
orientations for staircase events (step and ray events) defined in
previous section and the orientations at which two vertices of
$\vis(u)$ are aligned horizontally. The set of the orientations of the
latter type is to capture the event orientations at which a tip of
$\stc_\theta(u)$ is aligned horizontally with a tip of
$\stc_{\theta+\frac{\pi}{2}}(u)$. We call them \emph{double staircase
  events}.  We initialize $\mathcal{Q}$ with the latter type events.
Note that it does not increase the time and space complexities of the
event queue. 

\subsection{Computing \texorpdfstring{\lmrs}{LMRs} of type
  \texorpdfstring{$\typeb_1$}{B1}}
\label{sec:b1}
Consider a reflex vertex $t$ of $P$ that appears as a tip on
$\stc_\theta(u)$.  We use $f(t)$ to denote the upper tip of the step on
$\stc_{\theta+\frac{\pi}{2}}(u)$ aligned horizontally to $t$.  For
example, in Figure~\ref{fig:LMR_B_vvv}(a,b), $v=f(q)$ in
$\stc_{\theta+\frac{\pi}{2}}(u)$.  During the rotation of the coordinate axes, we
consider the change of $f(t)$ for each tip $t$ on
$\stc_\theta(u)$, as well as step, ray, and shift events on the double staircase.
At each orientation, $f(t)$ can be computed in $O(\log n)$ time 
via binary search on $S_{\theta+\frac{\pi}{2}}(u)$ with $t_y$ 
since the staircase chain is monotone with respect to the $y$-axis. 
Thus we do not need to save the value $f(t)$ for each tip $t$.
We consider the orientation when a tip $t$
on $\stc_\theta(u)$ and $f(t)$ on $\stc_{\theta+\frac{\pi}{2}}(u)$ are
aligned horizontally so that $f(t)$ is set to the next tip on
$\stc_{\theta+\frac{\pi}{2}}(u)$.  
At such a orientation, we detect a \dsc
candidate of the type $\typeb_1$ with top, left, bottom, and right
\sidec{s} as $u$, $p$, $q$, and $v=f(q)$, respectively. Note that
the bottom \sidec $q$ might have 
$q_x > u_x$, that is, $q$ might
appear as a tip on $\stc_{\theta+\frac{\pi}{2}}(u)$.
We process only the case that $q$ is a tip on $\stc_\theta(u)$, since
the other case can be handled when fixing $p$ as an upper side
contact, as described in the following, when step $(q,v)$
disappears by a step event on $\stc_{\theta+\frac{\pi}{2}}(p)$ and $u$ is a
tip on $\stc_{\theta+\pi}(p)$.

\begin{figure}[t]
  \centering \includegraphics[width=.8\textwidth]{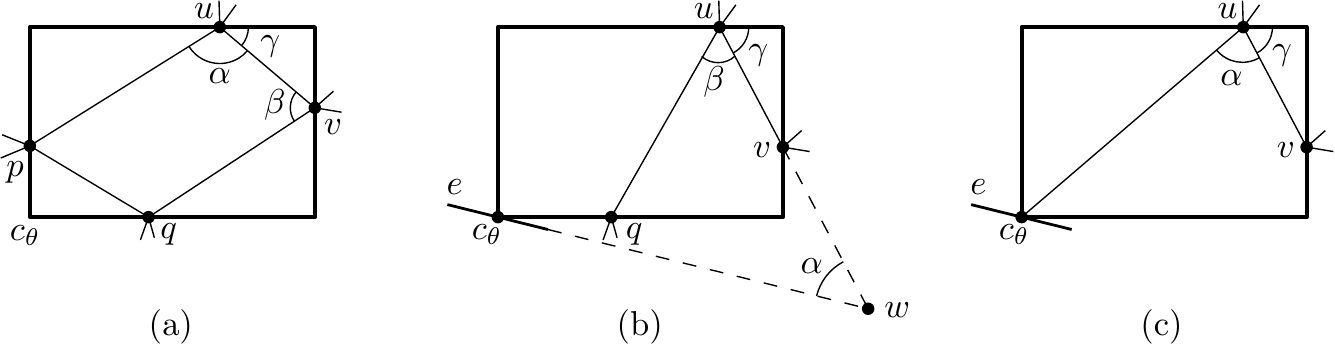}
  \caption{\lmrs of (a) 
    type $\typeb_1$, (b) type $\typeb_2$, and (c) type
    $\typeb_3$.}
  \label{fig:LMR_B_vvv}
\end{figure}

Consider an event $E$ occurring at $\theta$. When
a step of the double staircase that possibly contributes to a \dsc of
type $\typeb_1$ changes due to $E$, we detect possible \dsc{s}
that have been associated with it. Consider a \dsc $\{u, p, q, v\}$ as 
in Figure~\ref{fig:LMR_B_vvv}(a). If $E$ is a step
or shift event on $\stc_{\theta}(u)$, $O(1)$ tips appear or disappear
on $\stc_{\theta}(u)$ and $O(1)$ \dscs are detected 
at each such event.
If $E$ is a step or shift event on $\stc_{\theta+\frac{\pi}{2}}(u)$,
there are $O(n)$ tips $t$ on
$\stc_{\theta}(u)$ such that $f(t)$ changes.  Observe that such an
event corresponds to a step event on the staircase of $q$ in $\mathsf{C}_{\theta+\pi}$.
See Figure~\ref{fig:Bdistr}.
Thus, we may consider only step and
shift events on $\stc_\theta(u)$ together with the double staircase
events to detect possible \dsc{s} of type $\typeb_1$.

\begin{figure}[ht]
  \centering \includegraphics[width=0.5\textwidth]{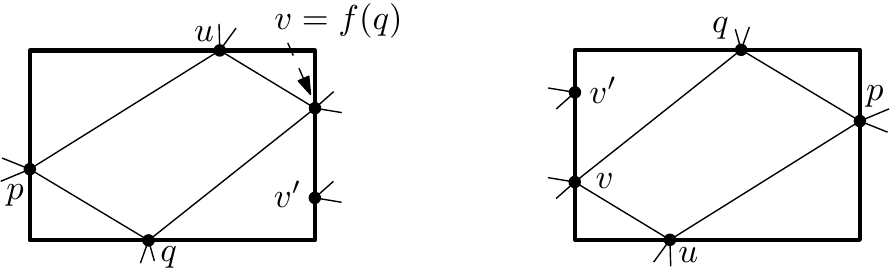}
  \caption{Step event on $\stc_{\theta+\frac{\pi}{2}}(u)$ that $f(q)$
    changes from $v$ to $v'$ (left) which is a step event on $\stc_{\theta+\pi}(q)$
    (right).}
  \label{fig:Bdistr}
\end{figure}
 
We consider $O(1)$ \dsc{s} for each event.  Let $Z=\{u,p,q,v\}$ be
a \dsc of type $\typeb_1$ such that $(p,q)$ is changed or $q$ and $v$
are aligned horizontally by an event $E$ at $\theta_E$.
Given a closed interval $J$, we can compute the set
$\Theta_Z$ of orientations $\theta_Z \in J$ maximizing
$\mu(\gth(Z))$ locally in $O(1)$ time because the area function $\mu(\gth(Z))=(|uv|\cos \gamma + |up| \cos (\pi-(\alpha+\gamma)))
                  (|uv|\sin \gamma + |qv| \cos(\frac{\pi}{2}-(\beta-\gamma)))$ 
has $O(1)$ extremal values in $J$. For angles $\alpha, \beta, \gamma$, see
Figure~\ref{fig:LMR_B_vvv}(a). 

We find the maximal interval $J$ of orientations for $Z = \{u,p,q,v\}$ 
found in an event $E$ occurring at $\theta_E$ such that 
$\theta_E\in J$ and all elements of $Z$ appear on the double staircase of $u$. 
This can be done by maintaining the latest orientation ($<\theta_E$)
at which $(p,q)$ starts to appear as a step, the latest orientation ($<\theta_E$) at which
$v$ starts to appear as a tip to the double staircase, and the orientation at which $f(q)$ was set to $v$.
Then $J=[\theta_a, \theta_E]$, where $\theta_a$ is the latest of orientation
at which all elements of $Z$ and the step consisting of elements of $Z$ start to appear on the double staircase while satisfying $f(q)=v$. 
Note that the \lmr{s} with contact $Z$ occur at every orientation of $\Theta_Z$ 
and the two endpoints (orientations) of $J$. Observe that the rectangle
$\Gamma_{\theta_Z}(Z)$ with $\theta_Z \in \Theta_Z$ corresponds to a
maximal configuration, and $\Gamma_{\theta}(Z)$ with $\theta$ being an
endpoint of $J$ corresponds to a breaking configuration. In this way,
we can compute $O(1)$ \lmrs satisfying $Z$ in $O(1)$ time.

\begin{lemma}
  \label{lem:double_staircase}
  Once the event map is constructed, for a fixed reflex vertex $u$, it
  takes $O(n^2\log n)$ time to maintain the double staircase of $u$
  and the event queue over all events during the rotation, and 
  to compute the \lmr{s} of type $\typeb_1$ having $u$ as the top \sidec.
\end{lemma}
For a reflex vertex $u$, we maintain an event queue $\mathcal{Q}$.  For
each event $E$ in $\mathcal{Q}$, our algorithm finds $O(1)$ \dsc{s}
$Z$ that become infeasible by $E$, and computes the \lmr{s} in $O(1)$
time.  Observe that a \dsc $Z$ of type $\typeb_1$
becomes infeasible only at shift, step, and double staircase events.
Since our algorithm is applied to every reflex vertex $u$ of $P$,
we do not need to process the shift and step events occurring on 
$\stc_{\theta+\frac{\pi}{2}}(u)$.
Therefore, we can detect every
possible \dsc $Z$ by processing the events in $Q$.  
By lemmas~\ref{lem:stair-complexity} and \ref{lem:double_staircase}, 
\begin{lemma}
  \label{lem:B1}
  Our algorithm computes all \lmr{s} of type $\typeb_1$ with largest
  area in $O(kn^2\log n)$ time, where $k$ is the number of reflex
  vertices.
\end{lemma}

\subsection{Computing \texorpdfstring{\lmrs}{LMRs} of type
  \texorpdfstring{$\typeb_2$}{B2}.}
\label{sec:b2}
A \dsc $Z=\{u,e,q,v\}$ of type $\typeb_2$ consists of three
reflex vertices $u, q, v$ realizing the top, bottom, right
\sidec and an oblique segment $e$ realizing the \cornerc at
the bottom-left corner $\blc$ of $\gth(Z)$. Let $w$ be the point where the
extended line of $e$ and the line through $u$ and $v$ cross.  If $w$
appears below $\blc$, then the area function is
$\mu(\gth(Z))=|uq|\sin(\beta+\gamma)\big(\cot(\gamma-\alpha)(|uw|\sin\gamma-|uq|\sin(\beta+\gamma))-|vw|\cos\gamma\big)$.
See Figure~\ref{fig:LMR_B_vvv}(b). The area of $\gth(Z)$ with $w$
appearing above $u$ can be computed in a similar way.  Note that there
are $O(1)$ orientations that maximize $\mu(\gth(Z))$ locally and they
can be computed in $O(1)$ time.

Observe that $q$ is contained in $\stc_{\theta}(u)$
or $(q,v)$ is a step on $\stc_{\theta+\frac{\pi}{2}}(u)$.  In addition
to the method from the Section~\ref{sec:b1}, we also handle the case
that $(q,v)$ is a step on $\stc_{\theta+\frac{\pi}{2}}(u)$ as follows.
For a reflex vertex $t$ appearing as a tip on
$\stc_{\theta+\frac{\pi}{2}}(u)$, let $g(t)$ be the edge that contains
$\feta(t)$.  We consider every change of $f(t)$ for each tip $t$ on
$\stc_{\theta}(u)$ and the every change of $g(t)$ for each tip $t$ on
$\stc_{\theta+\frac{\pi}{2}}(u)$ during the rotation.

Consider an event $E$ occurring at $\theta_E$. If $f(q)$ changes 
from $v$ to $v'$, we detect the \dsc $\{u, e, q, v\}$, where $e$ is
the edge containing the oblique segment of the step $q$ belongs to, 
in a similar way as we process such an event of type $\typeb_1$.
For a step or shift event on $\stc_{\theta}(u)$, and a
double staircase event, there exist only $O(1)$ \dsc{s} becoming
infeasible caused by the change of $f(q)$, and their \lmr{s} are computed
in $O(1)$ time.  In addition, at each event associated with
disappearance of a step $(p,q)$ on $\stc_\theta(u)$, we process a \dsc
$\{u, p, q, e\}$, where $e$ is the edge that contains $\flambda(u)$.
When $\flambda(u)$ no longer meets $e$, we process a \dsc
$\{u, p, q, e\}$ for each step $(p, q)$ on $\stc_\theta(u)$ as well.
For a step or shift event on $\stc_{\theta+\frac{\pi}{2}}(u)$, there
can be $O(n)$ tips $q$ on $\stc_\theta(u)$
such that $f(q)$ gets changed, but such a case overlaps with the case
of step or shift event on $\stc_{\theta+\pi}(q)$ with $\flambda(q)$
on $e$, which is handled for the double staircase of $q$.
Thus we do not handle them for the double staircase of $u$.

So it remains to consider the case for an event $E$ that changes
$g(q)$ for each tip $q$ on $\stc_{\theta + \frac{\pi}{2}}(u)$ at $\theta_E$.
Observe that $g(q)$ changes only if a double staircase event occurs
associated with $q$.  Whenever a new step $(q,v)$ appears on
$\stc_{\theta+\frac{\pi}{2}}(u)$, we do binary search for
$\feta(q)\in e$ in $\vis(q)$. When $g(q)$ changes or a step $(q,v)$
disappears caused by a step or a shift event on
$\stc_{\theta+\frac{\pi}{2}}(u)$, we find $O(1)$ \lmr{s} with \dsc
$Z=\{u, e, q, v\}$, and check if they are in $P$ by checking if the
boundary of the rectangles are in $P$, since we do not know if $e$
appears on $\stc_\theta(u)$ or not. These \lmr{s} can be computed 
in a similar way as we do for type $\typeb_1$.

We find the maximal interval $J$ of orientations for $Z=\{u,e,q,v\}$ 
such that $\theta_E\in J$ and all elements of $Z$ appear on the double staircase of $u$.
We can compute the set
$\Theta_Z$ of orientations $\theta_Z \in J$ that maximize
$\mu(\gth(Z))$ locally in $O(1)$ time because the area function 
has $O(1)$ extremal values in $I$.
Then for each $\theta_Z \in \Theta_Z$, we
do binary search in $C(q,v)$ within $J$ for $Z$, to
get two feasible orientations $\theta_1,\theta_2\in J$ closest to $\theta_Z$
with $\theta_1 \le \theta_Z$ and $\theta_2 \ge \theta_Z$ if they exist. Together with
the endpoints of $J$, we get the \lmr{s} with \dsc $Z$.  Then we
consider $O(1)$ \dsc{s} $Z$ of type $\typeb_2$ for each event and
their feasible orientations that maximize the area of $\gth(Z)$ can be
computed in $O(\log n)$ time. 
This way we can detect every \bc including $Z$ as well. See Figure~\ref{fig:BCs}.
 By capturing the changes of $f(q)$
and $g(q)$ together with the changes on the double staircase, we detect all possible \dsc{s} of type $\typeb_2$ and associate \bc{s}. 
Suppose that an \lmr appears as a $\gth(Z)$ for $Z=\{u,e,q,v\}$. 
For example, $e$ is on $\stc_\theta(u)$ and $q$ is on $\stc_\theta(u)$ with $f(q)=v$ or on $\stc_{\theta+\frac{\pi}{2}}(u)$ with $g(q)=e$. As we rotate the coordinate system, $Z$ becomes infeasible by the change of the double staircase, $f(q)$, or $g(q)$, say at event $E$.
However, the remaining conditions are not affected by $E$, so that we can detect $Z$ by our algorithm.
Together
with Lemma~\ref{lem:stair-complexity}, we have the following
lemma.
\begin{lemma}\label{lem:B2}
  Our algorithm computes all \lmr{s} of $\typeb_2$ with the largest
  area in $O(kn^2\log n)$ time, where $k$ is the number of reflex
  vertices of $P$.
\end{lemma}

\subsection{Computing \texorpdfstring{\lmrs}{LMRs} of type
  \texorpdfstring{$\typeb_3$}{B3}.}
Consider the case when \dsc $Z = \{u,e,v\}$ is feasible, where $e$ is a polygon
edge that appears as an oblique segment on $\stc_\theta(u)$ and $v$
is a tip on $\stc_{\theta+\frac{\pi}{2}}(u)$ (type $\typeb_3$). See
Figure~\ref{fig:LMR_B_vvv}(c).  To achieve the largest area, we
observe that the bottom-left corner of \lmr{s} satisfying $Z$ must lie at
the midpoint $c$ of the extended line segment $pq$ of $e$, where $p$ and
$q$ are the intersection points of the line containing $e$ with
$\eta(u)$ and $\delta(v)$, respectively.  The area function of
$\gth(\{u, \blc, v\})$, the rectangle with top \sidec on $u$,
bottom-left \cornerc on $\blc$, and right \sidec on $v$, is
convex with respect to $\blc \in l$, where $l$ is the line containing
$e$. If $c$ does not lie on $e$, $\blc$ must lie on a
point of $e$ closest to $c$ to maximize the rectangle area.

The area function of $\gth(Z)$ for a \dsc $Z$ of $\typeb_3$ is
$\mu(\gth(Z))={|uc_\theta|}\sin(\alpha+\gamma)\big(|uc_\theta|\cos(\pi-(\alpha+\gamma))+|uv|\cos\gamma\big)$,
where $c_\theta$ is the midpoint of $pq$ (if the midpoint lies on $e$) or the
endpoint of $e$ that is closer to the midpoint (otherwise) at $\theta$. Since the midpoint 
moves along $l$ in one direction as $\theta$ increases, 
there are $O(1)$ intervals of orientations at which the midpoint of $pq$ is contained in $e$, and thus this area function has $O(1)$ extremal values in $I$.

Our algorithm for computing all \lmr{s} of type $\typeb_3$ is simple.
First we fix the top \sidec on $u$. For each pair of an edge $e$
and a reflex vertex $v$, we compute the set $\Theta_Z$ 
of orientations that maximize $\mu(\gth(Z))$ locally for $Z = \{u,e,v\}$.
Observe that $\Theta_Z$ consists of $O(1)$
orientations because the area function has $O(1)$ extremal values.
Then 
for each $\theta_Z \in \Theta_Z$, we find two orientations $\theta_1,\theta_2$ closest to $\theta_Z$
with $\theta_1 \le \theta_Z$ and $\theta_2 \ge \theta_Z$ such that
the top-right corner of $\gth(Z)$ is contained in $P$ by applying
binary searching on $C(u,v)$.
Finally, we check if
$\gth(Z)$ is contained in $P$ for $O(1)$ such orientations $\theta$ by checking if the
boundary of the rectangles are contained in $P$. 
This way we can
compute all \lmr{s} of $\typeb_3$ with top \sidec on $u$. See
Figure~\ref{fig:BCs}.  By using the event map and
Lemma~\ref{lem:stair-complexity}, we have the following lemma.

\begin{lemma}\label{lem:B3}
  Our algorithm computes all \lmr{s} of $\typeb_3$ with largest area
  in $O(kn^2\log n)$ time, where $k$ is the number of reflex vertices
  of $P$.
\end{lemma}


\section{Computing a largest rectangle of types \typec and \typed}
\label{sec:lmrs-type-CD}
\lmr{s} of types \typec and \typed can be computed in a way similar to
the one for type \typeb. For each reflex vertex $u$, we find all
\lmr{s} of types \typec and \typed that have $u$
on its top side while maintaining the double staircase of $u$.

Let an edge $e$ be an element of a \dsc of type \typec or \typed that
is a \cornerc on the top-right corner of an \lmr.
Note that $e$ contains $\flambda(u)$ at some $\theta$.
\begin{figure}[!ht]
  \centering
  \includegraphics[width=\textwidth]{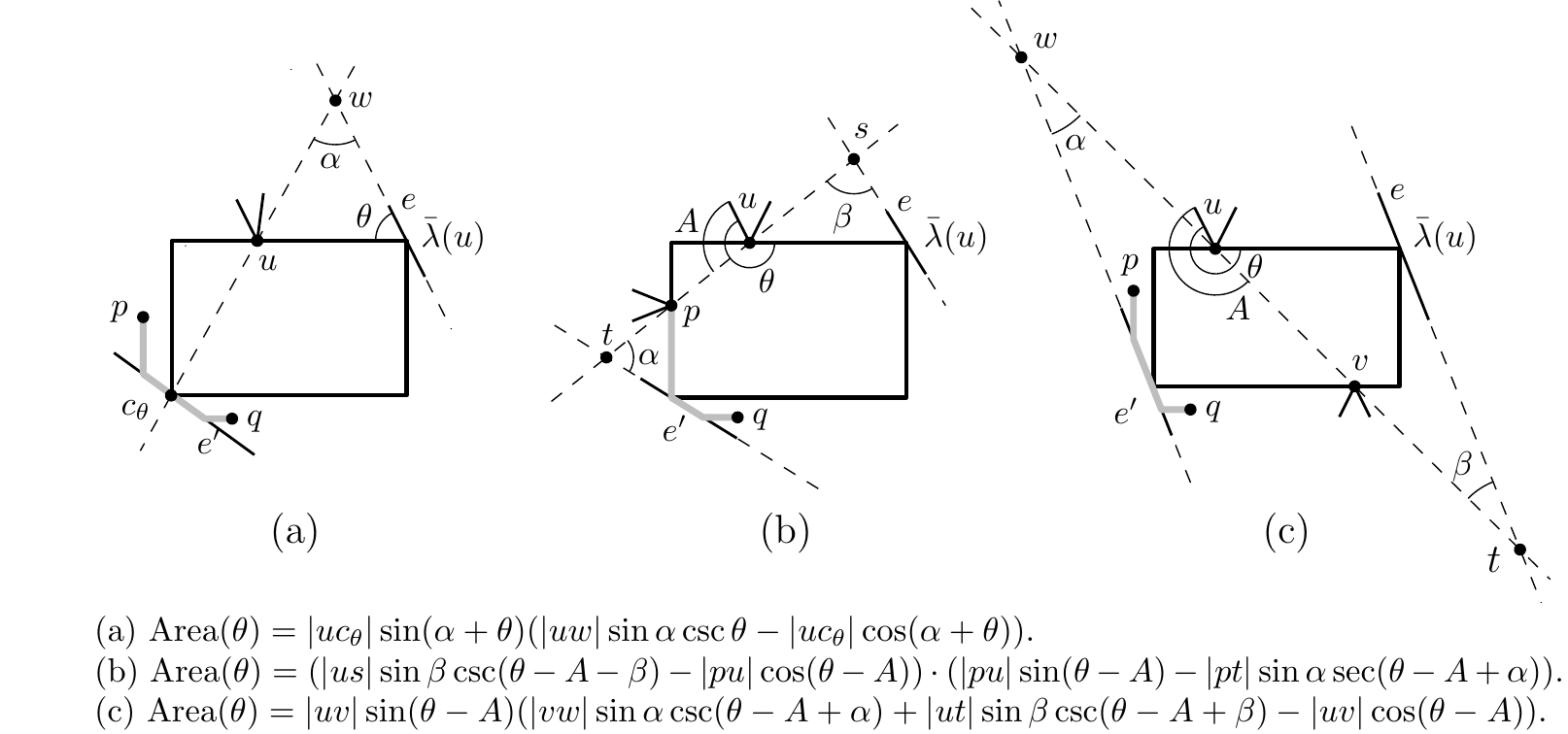}
  \caption{Computing the area of some \lmr{s} of type \typec. (a) Type $\typec_1$, (b) type $\typec_2$, and (c) type $\typec_3$.}
  \label{fig:lmr-C}
\end{figure}

\subsection{Computing \lmrs of type \typec}
As we do for type $\typeb$ in the Section~\ref{sec:comp-b}, 
we detect the events at which a possible \dsc $Z$ of type \typec
becomes infeasible, and compute the \lmrs satisfying $Z$.
We compute \lmr{s} of type \typec at (1) every step and shift event (and ray event) on
$\stc_\theta(u)$ such that a step 
on $\stc_\theta(u)$ changes,
(2) every ray event such that $\flambda(u)$ meets a vertex of $P$,
(3) every step event of the step incident to $\flambda(u)$ 
on $\stc_{\theta+\frac{\pi}{2}}(u)$, and
(4) every event such that the lower tip $q$ of step $(p,q)$ on $\stc_\theta(u)$ and 
the tip $v$ of step $(\flambda(u),v)$ on $\stc_{\theta+\frac{\pi}{2}}(u)$ are aligned horizontally. 

In case (1) such that a step $(p,q)$ on $\stc_\theta(u)$ changes, 
we consider three subcases: \dsc{s} $\{u,e',e\}$ of type 
$\typec_1$ (Figure \ref{fig:lmr-C}(a)),
\dsc{s}
$\{u,p,e',e\}$ of type 
$\typec_2$ (Figure \ref{fig:lmr-C}(b)),
and \dsc{s} $\{u,e,q,e'\}$ of type $\typec_3$ for edge
$e'$ containing 
$\ob(p,q)$ (Figure \ref{fig:lmr-C}(c)).
In case (2), $\flambda(u)$ no longer lies on
$e$. Thus we consider \dsc{s} $\{u,e',e\}$ (type $\typec_1$), 
\dsc{s} $\{u,p,e',e\}$ and $\{u,e',q,e\}$
(type $\typec_2$ and $\typec_3$), and 
\dsc{s} $\{u,e',v,e\}$ (type $\typec_3$) for each edge $e'$ containing
an oblique segment $\ob(p,q)$ on $\stc_\theta(u)$ and the tip $v$ of step $(\flambda(u),v)$ on
$\stc_{\theta+\frac{\pi}{2}}(u)$.
In case (3) such that step $(\flambda(u),v)$ changes,
we consider three \dscs $\{u,e',e\}$, $\{u,p,e',e\}$ and $\{u,e',v,e\}$ 
for edge $e'$ containing $\ob(p,q)$, 
where $(p,q)$ is the step on $\stc_\theta(u)$ intersecting $\eta(v)$.
If $\fdelta(\flambda(u))$ meets a reflex vertex $v'$, 
we consider \dscs $\{u,e',e\}$ and $\{u,p,e',e\}$
for edge $e'$ containing $\ob(p,q)$, 
where $(p,q)$ is the step on $\stc_\theta(u)$
intersecting $\eta(v')$.
For a step event of step incident to $\flambda(u)$,
there can be $O(n)$ steps $(p,q)$ on $\stc_\theta(u)$ such that
\dscs $\{u,e'e\}, \{u, p,e',e\}$ and $\{u,e',q,e\}$ 
for $e'$ containing $\ob(p,q)$ become infeasible.
These changes correspond to the changes caused by the step event
that $e$ disappears from $\stc_{\theta+\pi}(p)$ or $\stc_{\theta+\pi}(q)$ .
Thus, they are handled for the double staircase of $p$ or $q$.
In case (4), we consider \dsc with $\{u,e',v,e\}$ of type $\typec_3$
for step $(\flambda(u),v)$ on $\stc_{\theta+\frac{\pi}{2}}(u)$ and
 edge $e'$ containing $\ob(p,q)$.


For each \dsc $Z$ of type $\typec$, we have its possibly feasible
interval by maintaining the orientations where each step had
the last step event associated with it, $\flambda(u)$ met another
polygon vertex, and a tip of $\stc_\theta(u)$ and the lower tip $v$ of $(\flambda(u),v)$ 
were aligned horizontally.
Within such a feasible interval of a \dsc $Z$, we compute all \lmr{s}
with contact $Z$. It is not difficult to see that the area functions of
type $\typec_2$ and $\typec_3$ depend only on $\theta$.
And as in type $\typeb_3$, for a fixed $\theta$ and 
$\gth(Z)$ of a \dsc $Z$ of type $\typec_1$, the bottom-left corner $\blc$ of $\gth(Z)$
on $e'$ is the midpoint of $p=l \cap \eta(\flambda(u))$ and
$q= l \cap \delta(\flambda(u))$, where $l$ is the extended line of
$e'$. If the midpoint
does not appear on the $\stc_\theta(u)$,
we simply take the point on the oblique
segment contained in $e'$ that is closest to the midpoint.
This is because $\mu(\gth(\{u,\blc,\flambda(u)\}))$ is convex
on the position of $\blc$ on $l$. 
Since the number of closed intervals where the midpoint appears on $\stc_\theta(u)$
is $O(1)$, the number of \lmr{s} with contact $Z$ is $O(1)$.
As a result, for each detected \dsc of type
$\typec$, we compute the set of $O(1)$ orientations achieving
\lmr{s} in $O(1)$ time. 

We detect every \bc of type $\typec_2$ with an additional top-left \cornerc 
in case (1). We detect every \bc of type $\typec_3$ with an additional bottom \sidec 
in cases (1) and (3-4) since there are two bottom \sidec{s} (vertices) aligned horizontally in such a \bc. 
The other \bcs of type \typec with an additional top \sidec are detected 
in case (2).
\begin{figure}[!ht]
  \centering
  \includegraphics[width=\textwidth]{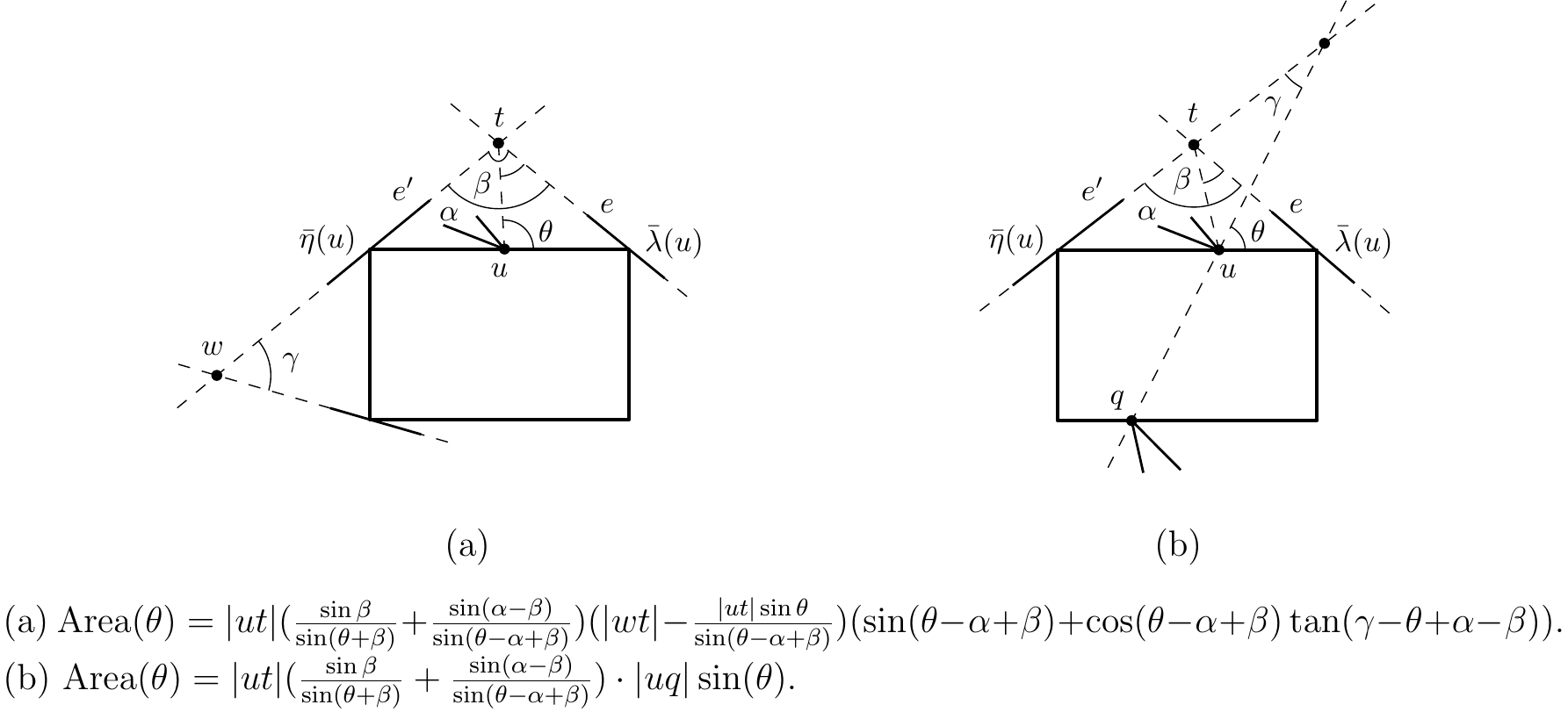}
  \caption{Area functions of type \typed. (a) Type $\typed_1$ and (b) type $\typed_2$.}
  \label{fig:LMR_D}
\end{figure}
\subsection{Computing \lmrs of type \typed}
Type \typed is almost the same as type \typec, except that
the base subset of \dsc is $\{u, e', e\}$,
where $e'$ and $e$ are polygon edges containing $\feta(u)$ and
$\flambda(u)$, respectively. Moreover, instead of the bottom-left
corner contacts on oblique segments of $\stc_\theta(u)$ considered in type $\typec$,
we simply consider corner contacts on the oblique segments of the step incident to $\feta(u)$ on
$\stc_\theta(u)$ and the step incident to $\flambda(u)$ on
$\stc_{\theta+\frac{\pi}{2}}(u)$ (Figure~\ref{fig:LMR_D}(a))
or
the side contacts on the bottom side (Figure~\ref{fig:LMR_D}(b)).

Therefore, we detect \dsc{s} and \bc{s} of types $\typed_1$ and $\typed_2$ corresponding to the events at which
$\feta(u)$ or $\flambda(u)$ meets a polygon vertex,
and to the step events on the step incident to $\feta(u)$ on $\stc_\theta(u)$
or on the step events on the step incident to $\flambda(u)$ on $\stc_{\theta+\frac{\pi}{2}}(u)$, and 
to the events at which the lower tip $q$ of step $(\feta(u),q)$ and
the lower tip $v$ of step $(\flambda(u),v)$ are aligned horizontally.
There are \bc{s} satisfying $\{u,e',e_l,t,e\}$ of type $\typed_2$,
where $t$ is a tip of $\stc_{\theta+\frac{\pi}{2}}(u)$ and $e_l$ is the edge containing $\fdelta(\feta(u))$ 
(the last \bc of type $\typed_2$ in Figure~\ref{fig:BCs}).
We can detect such \bc{s} 
of type $\typed_2$ that has bottom-left \cornerc on an edge $e_l$ 
in addition to the contacts $Z=\{u,e',t,e\}$ of type $\typed_2$
whenever we detect \dsc $Z$ in $O(1)$ time
since $e_l$ contains an oblique segment of the step incident to $\feta(u)$.
The \lmr{s} of such a \bc can be computed in $O(1)$ time by solving basic system of linear equations.


There is another \bc type of type $\typed_1$ that has
bottom-right \cornerc on an edge in addition to the contacts
of type $\typed_1$.
We will handle such \bc{s} when we handle
\dsc{s} of type $\typee_3$ 
in Section~\ref{sec:comp-E} (the last \bc of type $\typee_3$ in Figure~\ref{fig:BCs}).
Thus, we exclude this \bc type from the breaking configurations
of type $\typed$.

\begin{lemma}\label{lem:CD}
  Our algorithm computes all \lmr{s} of types \typec and \typed. 
\end{lemma}
\begin{proof}
By Lemma~\ref{lem:configurations}, an \lmr satisfying a \dsc
  $Z$ is in a maximal or breaking configuration. The
  \dsc{s} and \bc{s} of type \typec get infeasible via one of the cases (1-4) and
  therefore they are detected by our algorithm.
  Also, the \dsc{s} and \bcs of type \typed are handled by the step events
  on the step incident to $\feta(u)$ or $\flambda(u)$ of the double staircase of $u$.
\end{proof}
%
\begin{lemma}
  We can compute a largest rectangle among all \lmrs of types
  \typec and \typed in $O(kn^2 \log n)$ time using $O(kn^2 )$ space, where $k$
  is the number of reflex vertices of $P$. 
\end{lemma}
\begin{proof}
  The event map can be constructed in
    $O(kn^2 \log n)$ time using $O(kn^2)$ space by
    Lemma~\ref{lem:stair-eventmap}.
 Then for each reflex
  vertex, we construct its double staircase
  and maintain it in $O(n^2\log n)$ time using $O(n)$ space by Lemma~\ref{lem:stair-update}.
  
  Note that for each \dsc or \bc $Z$ of types \typec and \typed it takes $O(\log n)$ time to compute $O(1)$ 
  \lmrs and check their feasibility.
  For each reflex vertex $u$ as the top \sidec, there are $O(n^2)$ events, and each event yields $O(1)$ 
  \dsc{s} or \bc{s} in cases (1), (3) and (4).
  Since there are at most $O(n)$ steps in $\stc_\theta(u)$
  and $O(n)$ ray events in total, the total number of \dsc{s} handled
  at case (2) is $O(n^2)$ for a reflex vertex $u$ as the top \sidec.
  Therefore, for each reflex vertex $u$, we can compute \lmrs of types \typec and \typed in $O(n^2 \log n)$ time, by Lemma~\ref{lem:CD}, which implies the theorem.
\end{proof}
\section{Computing a largest rectangle of type E}
\label{sec:comp-E}
We consider the \lmrs of type \typee.  Let $u$ be a reflex vertex of $P$.
We detect every \dsc $Z$ of type \typee, containing $\{u, e_l, e_r\}$,
where $u$ is the top \sidec, $e_l$ the bottom-left \cornerc,
and $e_r$ the bottom-right \cornerc.  Observe that for each \lmr
satisfying $Z$, $e_l$ and $e_r$ appear as oblique segments
$\ob(p,q) \subset e_l$ and $\ob(t,v) \subset e_r$ of $\stc_\theta(u)$
and $\stc_{\theta+\frac{\pi}{2}}(u)$, respectively, such that
$\feta(t) \in \ob(p,q)$ or $\flambda(q) \in \ob(t,v)$, depending on
whether $q_y \le t_y$ or not.  Using this fact, we detect the events
at which $Z$ becomes infeasible, and compute the \lmrs satisfying $Z$.

We compute \lmr{s} of type \typee at (1) every step and shift event
(and ray event) with a step containing an oblique segment on the
double staircase, and (2) every event such that $\flambda(q)$ meets
$\fdelta(t)$ on an edge $e$ for a tip $q$ of $\stc_\theta(u)$ and a
tip $t$ of $\stc_{\theta+\frac{\pi}{2}}(u)$
$\stc_{\theta+\frac{\pi}{2}}(u)$ (Figure~\ref{fig:LMR_E_ref}(d).

In case (1),
at a step or shift event with a step $(p,q)$ on
$\stc_\theta(u)$ such that $\ob(p,q)\subset e_l$ is nonempty and
$\flambda(q)$ is contained in an oblique segment
$\ob(t,v)\subset e_r$ of $\stc_{\theta+\frac{\pi}{2}}(u)$, we consider
\dsc{s} $Z_1=\{u,e_l,e_r\}$ 
(Figure~\ref{fig:LMR_E_ref}(a)) 
and 
$Z_2=\{u,p,e_l,e_r\}$ (Figure~\ref{fig:LMR_E_ref}(b)),
together with
their corresponding \bcs, for edges $e_l$ and $e_r$ of $P$.  A step or
shift event with a step on $\stc_{\theta+\frac{\pi}{2}}(u)$ can be
handled in a symmetric way.  If $e_l$ appears as the oblique
segment of the step incident to $\feta(u)$ of $\stc_\theta(u)$, we
let $Z_2=\{u,e',e_l,e_r\}$, where $e'$ is the edge containing
$\feta(u)$. 
See Figure~\ref{fig:LMR_E_ref}(c).

At an event $E$ of case (2) occurring at $\theta_E$, we
have a step $(p,q)$ on $\stc_\theta(u)$ and a step $(t,v)$ on
$\stc_{\theta+\frac{\pi}{2}}(u)$ such that $\flambda(q)$ meets
$\fdelta(v)$ on an edge of $P$.  We consider the same \dsc{s} $Z_1$
and $Z_2$ considered in case (1).  Observe that $E$ corresponds to the
step event of the double staircase of $v$ at $\theta_E-\frac{\pi}{2}$.
The double staircase of $v$ has a step event at 
$\theta_E-\frac{\pi}{2}$ that $\fdelta(\flambda(v))$ meets $q$
(equivalently, $\flambda(q)$ meets $\fdelta(v)$ on an edge of $P$ at
 $\theta_E$). See Figure~\ref{fig:LMR_E_ref}(d).
 
 Thus, we
can capture $E$ by maintaining the double staircase of $v$ and insert
$E$ to the event queue of $u$ in $O(\log n)$ time.  We compute this
type of events for all reflex vertices of $P$ by maintaining double
staircases of the reflex vertices of $P$ and insert the events to the
event queues of their corresponding reflex vertices whenever such
events are found. There are $O(kn^2)$ events in total, and they can be found and inserted to the event
queues in $O(kn^2\log n)$ time.

Whenever detecting a \dsc $Z$, we take a closed interval $J$ of
orientations at which $Z$ is possibly feasible, and compute $O(1)$
\lmr{s} with contact $Z$ within $J$. When $Z$ contains
$\{u, e_l, e_r\}$ as the top \sidec $u$, bottom-left \cornerc
$e_l$, and bottom-right \cornerc $e_r$, $J$ is the interval such
that $\ob(p,q) \subseteq e_l$, $\ob(t,v) \subseteq e_r$, and
$\flambda(\fdelta(p))\in \ob(t,v)$ or $\flambda(q) \in \ob(t,v)$.
Note that the interval satisfying $\flambda(\fdelta(p))\in \ob(t,v)$
or $\flambda(q) \in \ob(t,v)$ can be computed in $O(\log n)$ time
using binary search on $L(p,e_l)$ and $L(v,e_r)$.  If $Z$ contains $p$
as the left \sidec or $e'$ as the top-left \cornerc, we
consider the \bc{s} such that $v$ is the right \sidec or 
there is another \cornerc on the top-right corner. 
Note that the \bc of the second case corresponds to a \bc of type $\typed_1$. 
The orientation at which such a \bc occurs can be computed in $O(1)$ time
by solving basic system of linear equations. Therefore, $J$ can be computed in
$O(\log n)$ time.

\begin{figure}[!ht]
  \centering \includegraphics[width=\textwidth]{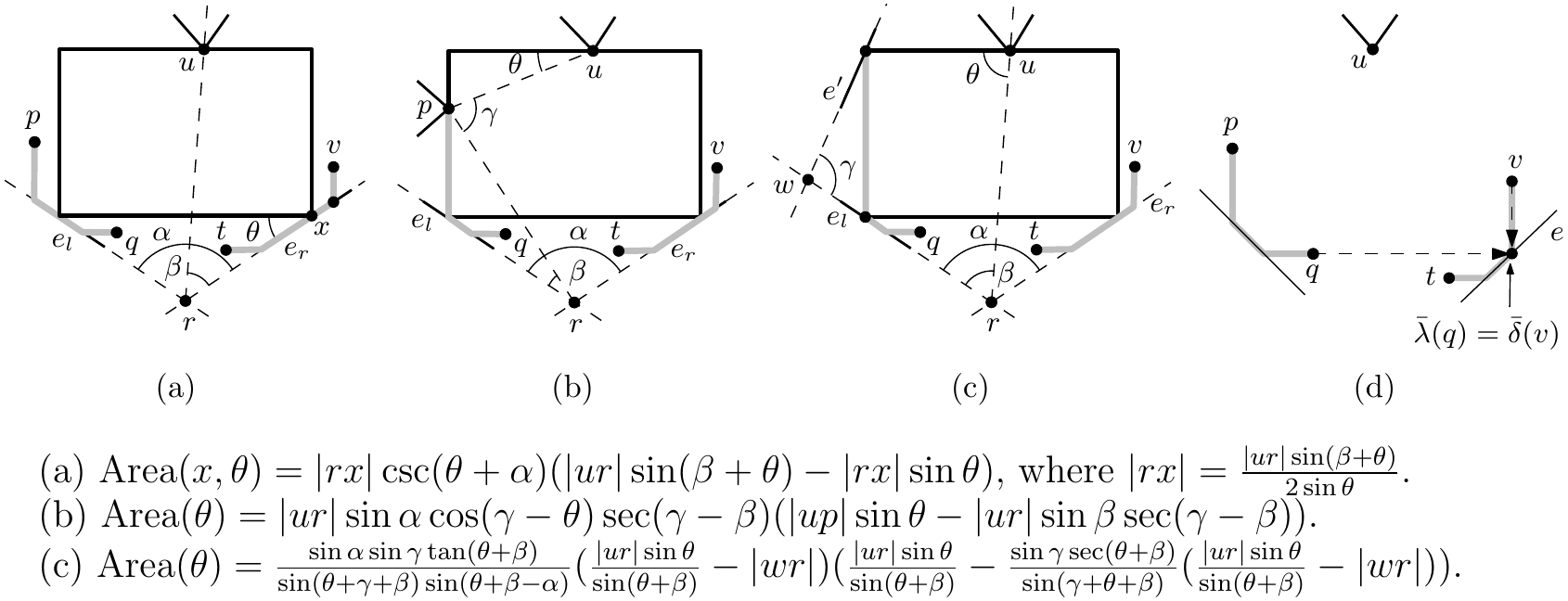}
  \caption{(a-c) \lmrs of type \typee, with reflex vertex $u$ on the
    top side, and their area functions. (d) An event at 
    $\theta_E$ at which $\flambda(q)$ meets $\fdelta(v)$ on $e$ for tips $q$ and
    $v$ and edge $e$.}
  \label{fig:LMR_E_ref}
\end{figure}

The area functions of some cases of type \typee are given in Figure
\ref{fig:LMR_E_ref} (a-c). The area functions of the other cases of
type \typee can be defined in a similar way.
The area functions of type \typee have
$O(1)$ extremal values. 

\begin{lemma}
  \label{lem:E}
  Our algorithm computes all \lmr{s} of type \typee.
\end{lemma}
\begin{proof}
  Let $Z$ be a \dsc of type \typee such that $Z$ is feasible at some
  $\theta$.  Without loss of generality, we may assume that
  $\{u, e_l, e_r\} \subset Z$, where $u$ is the top \sidec,
  $e_l$ the bottom-left \cornerc, and $e_r$ the bottom-right
  \cornerc. We also assume $\ob(p,q) \subset e_l$ for some step
  $(p,q)$ of $\stc_\theta(u)$, and $\ob(t,v) \subset e_r$ for some
  step $(t,v)$ of $\stc_{\theta+\frac{\pi}{2}}(u)$.

  If $Z=\{u,e_l,e_r\}$ (type $\typee_1$), $Z$ becomes infeasible when
  $\ob(p,q)$ or $\ob(t,v)$ disappears from the double staircase of $u$
  or $\flambda(q)$ meets $\fdelta(v)$ on $e_r$.  We detect every
  event such that $\ob(p,q)$ or $\ob(t,v)$ disappears (case (1)) and
  every event such that $\flambda(q)$ meets $\fdelta(v)$ on $e_r$
  (case (2)).  If $Z$ is $\{u,p,e_l,e_r\}$ of type $\typee_2$, it
  becomes infeasible by an event that we detect for \dscs of type
  $\typee_1$.  or by the step $(p,q)$ that changes. We detect these
  events in cases (1) and (2) and detect the events such that step
  $(p,q)$ changes in case (1).  Similar to type $\typee_2$, \dsc
  $Z=\{u,e',e_l,e_r\}$ of type $\typee_3$ becomes infeasible by an
  event that we detect for \dscs of type $\typee_1$ or by a step
  and shift event of the step incident to $\feta(u)$. Such a step and
  shift event is detected in case (1).
   
  For each \dsc $Z$, we take a closed interval $J$ of orientations
  where $Z$ is possibly feasible, and therefore our algorithm computes all
  \lmr{s} of type \typee.
\end{proof}

We compute $O(1)$ \lmrs for each event and check if they are contained in $P$ 
in $O(\log n)$ time.  There are $O(kn^2)$ events
corresponding to case (2) which are computed in
$O(kn^2 \log n)$ time before we handle the events of type \typee.  By
Lemma~\ref{lem:stair-complex}, we have the following
lemma.
\begin{lemma}
  We can compute a largest rectangle among all \lmrs of type
  \typee in $O(kn^2 \log n)$ time using $O(kn^2 )$ space, where $k$ is
  the number of reflex vertices of $P$.
\end{lemma}

\section{Computing a largest rectangle of type F}
\label{sec:comp-F}

\begin{figure}[!ht]
\centering \includegraphics[width=\textwidth]{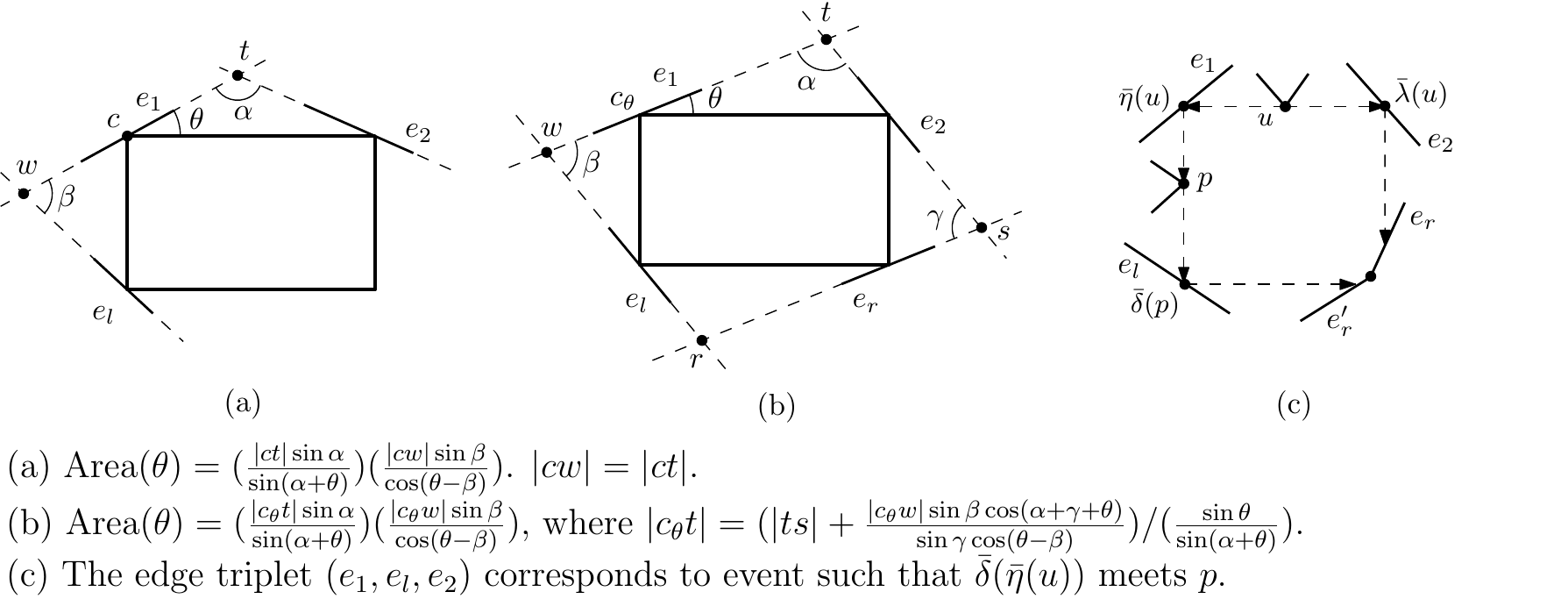}
\caption{\lmrs of type \typef.}\label{fig:LMR_F}
\end{figure}
To find all \lmrs of type \typef, we compute the maximal configurations 
and breaking configurations of \dsc{s} of type \typef as follows.
Consider a \dsc $Z_1=\{e_1,e_l,e_2\}$ of type $\typef_1$ (Figure \ref{fig:LMR_F}(a))
and a \dsc $Z_2=\{e_1,e_l,e_r,e_2\}$ of type $\typef_2$ (Figure \ref{fig:LMR_F}(b)).
Then the \lmr of a \bc of type $\typef_1$ is the rectangle satisfying $Z_1 \cup \{u\}$ or $Z_1 \cup \{e_r\}$,
or a rectangle satisfying $Z_1$ with $\cornerc$ on an end vertex of
an edge in $Z_1$,
where $u$ is a reflex vertex and $e_r$ is an edge of $P$.
The \lmr satisfying $Z_1 \cup \{u\}$ belongs to type $\typed_1$ or $\typee_3$ 
which is computed as an \lmr of 
\typed or \typee, by Lemma \ref{lem:CD} and \ref{lem:E}.
The \lmr satisfying $Z_1 \cup \{e_r\}$ belongs to type $\typef_2$ and
it is considered for type $\typef_2$.
The \lmr of a \bc of type $\typef_2$ is the rectangle satisfying $Z_2\cup \{u\}$, $Z_2\cup \{e'\}$
or the rectangle satisfying $Z_2$ with $\cornerc$ on an end vertex of an edge in
$Z_2$,
where $u$ is a reflex vertex and $e'$ is an edge of $P$.
The \lmr satisfying $Z_2\cup \{u\}$ belongs to a \bc of type $\typee_3$ which is
computed as an \lmr of type \typee by Lemma \ref{lem:E}.
(See the last \bc of type $\typee_3$ in Figure \ref{fig:BCs}.)
Thus, we have the following lemma.
\begin{lemma}
\label{lem:F_lmr}
Every \lmr is of a type in $\{\typea,\typeb,\typec,\typed,\typee\}$, a maximal configuration of type \typef, or a breaking configuration  $Z$ of type \typef containing a \cornerc on an end vertex of an edge in $Z$.
\end{lemma}

We say an edge pair $(e_1,e_2)$ is \textit{\textsf{h}-aligned} (and \textit{\textsf{v}-aligned}) 
  at $\theta$
  if there are points $p_1\in e_1$ and $p_2\in e_2$ such that
  $p_1p_2$ is horizontal (and vertical) and is contained in $P$ at $\theta$.
A pair $(e_1,e_2)$ of edges is \textit{\textsf{h}-misaligned} (and \textit{\textsf{v}-misaligned})  
at $\theta$ if the pair is not \textsf{h}-aligned (and not \textsf{v}-aligned) at $\theta$.
Note that a edge pair $(e_1,e_2)$ changes between being \textsf{h}- or
\textsf{v}-aligned and being \textsf{h}- or \textsf{v}-misaligned only
when two vertices of $P$ are aligned horizontally or vertically during the rotation.
We say a triplet $(e_1,e_l,e_2)$ of edges \textit{\textsf{t}-aligned}
at $\theta$ if there is a point $x \in e_1$ such that
 $\flambda(x) \in e_2$ and $\fdelta(x) \in e_l$ at some $\theta'\in\{\theta,\theta+\frac{\pi}{2},\theta+\pi,\theta+\frac{3\pi}{2}\}$.
 An edge triplet $(e_1,e_l,e_2)$ is \textit{\textsf{t}-misaligned}
if it is not \textsf{t}-aligned at $\theta$.  
 See Figure~\ref{fig:LMR_F}(c).
\begin{lemma}
An edge triplet $T=(e_1,e_l,e_2)$ becomes \textsf{t}-misaligned only if 
$(e_1,e_2)$ or $(e_1,e_l)$ becomes \textsf{h}- or \textsf{v}-misaligned,
or $\fdelta(\feta(u))$ meets $p$ for a vertex pair $(u,p)$ at $\theta$.
\end{lemma}
\begin{proof}
Without loss of generality, assume that $T$ is \textsf{t}-aligned at 
$\theta_1$ but becomes \textsf{t}-misaligned at $\theta_2$ for any 
$\theta_1\in[\theta_2-\eps,\theta_2)$ with small $\eps>0$. 
Then there is a point $w \in e_1$ such that $\flambda(w) \in e_2$ and 
$\fdelta(w) \in e_l$ at 
$\theta_1$. This implies that $(e_1,e_2)$ is \textsf{h}-aligned and 
$(e_1,e_l)$ is \textsf{v}-aligned at $\theta_1$. 

Assume to the contrary that $T$ becomes \textsf{t}-misaligned at $\theta_2$ but 
$(e_1,e_2)$ is \textsf{h}-aligned, $(e_1,e_l)$ is \textsf{v}-aligned 
and no vertex pair $(u,p)$ satisfies $\fdelta(\feta(u))$ meeting $p$ at $\theta$.
Let $X_\theta$ be the set of points $w\in e_1$ with 
$\flambda(w)\in e_2$, and let $Y_\theta$ be the set of points
$w\in e_1$ with $\fdelta(w)\in e_l$ at $\theta$. Observe that $X_\theta$ 
and $Y_\theta$ form line segments contained in $e_1$ and they change 
continuously during rotation from $\theta_1$ to $\theta_2$.
Since $T$ becomes \textsf{t}-misaligned at $\theta_2$, there is no point 
$w\in e_1$ such that $\flambda(w) \in e_2$ and 
$\fdelta(w) \in e_l$ at $\theta_2$, that is, $X_{\theta_2}\cap Y_{\theta_2}=\emptyset$.
If $X_{\theta_2}=\emptyset$, $(e_1,e_2)$ is \textsf{h}-misaligned. If $Y_{\theta_2}=\emptyset$, 
$(e_1,e_l)$ is \textsf{v}-misaligned. The only remaining case is that
$X_{\theta_2}\neq\emptyset$ and $Y_{\theta_2}\neq\emptyset$ but $X_{\theta_2}\cap Y_{\theta_2}=\emptyset$.
Since $T$ is \textsf{t}-aligned at $\theta_1$, that is, $X_{\theta_1}\cap Y_{\theta_1}\neq\emptyset$,
there must be a point $w\in e_1$
such that $\flambda(w)$ moves out of $e_2$ and $\fdelta(w)$ moves out of 
$e_l$ at $\theta_2$. This occurs when $\flambda(w)$ meets a vertex $u$ and 
$\fdelta(w)$ meets a vertex $p$, that is, $\fdelta(\feta(u))$ meets $p$.
This gives a contradiction.
\end{proof}
 
We compute \lmrs of type \typef at 
(1) every event such that two vertices are aligned horizontally or vertically, and 
(2) every event such that $\fdelta(\feta(u))$ meets $p$ for every vertex pair 
$(u,p)$ at 
$\theta'\in\{\theta,\theta+\frac{\pi}{2},\theta+\pi,\theta+\frac{3\pi}{2}\}$ (Figure \ref{fig:LMR_F}(c)).
 In case (1), at an event such that two vertices $u$ and $v$ are aligned horizontally,
 we find an edge pair $(e_1,e_2)$ which becomes \textsf{h}-misaligned in $O(\log n)$ time 
 using ray-shooting queries with $\eta(u)$ and $\lambda(v)$, assuming that $u_x<v_x$ if such pair exists.
 Then we also find edges $e_l$ and $e_r$ such that $e_l$ contains $\fdelta(\feta(u))$ and
 $e_r$ contains $\fdelta(\flambda(v))$. We can find such edges in $O(\log n)$ time 
 using ray-shooting queries.
Then we compute the set $\Theta_{Z_i}$ of orientations that maximize $\mu(\gth(Z_i))$
for each \dsc $Z_1=\{e_1,e_l,e_2\}$, $Z_2= \{e_1,e_r,e_2\}$ and $Z_3= \{e_1,e_l,e_r,e_2\}$ 
using the area functions in Figure~\ref{fig:LMR_F}
and check if $\gth(Z_i)$ is contained in $P$ for $\theta\in\Theta_{Z_i}$. 
Observe that the top-left \cornerc of every \lmr of type $\typef_1$ lies
at the midpoint $c$ of $wt$, for the intersection $w$ of two lines, one containing 
$e_1$ and one containing $e_l$, and the intersection $t$ of two lines,
one containing $e_1$ and one containing $e_2$.
(See Figure~\ref{fig:LMR_F} for $wt$.)
Note that every area function in Figure \ref{fig:LMR_F} has $O(1)$ extremal values.
If $c\not\in e_1$, we take the point on $e_1$ that is closest to the midpoint.
Thus, each $\Theta_{Z_i}$ has $O(1)$ elements and we can check 
for each $\gth(Z_i)$ if it is contained in $P$
in $O(\log n)$ using ray-shooting queries.
We also compute the \bc{s} satisfying $Z_i$ with $\cornerc$ on an end vertex 
of an edge in $Z_i$, and 
check their feasibility. There are $O(1)$ such \bc{s} which can be computed in $O(1)$ time. We can compute in $O(1)$ time $\mu(\gth(Z))$
for each \bc $Z$. 
 An event at which two vertices $u$ and $v$ are aligned vertically can be handled in a symmetric way.
 
 In case (2), when $\fdelta(\feta(u))$ meets $p$,
 we find an edge triplet $(e_1,e_l,e_2)$ which becomes \textsf{t}-misaligned
 in $O(\log n)$ time using ray-shooting queries with $\eta(u), \lambda(u)$ and $\delta(p)$.
 We also find edges $e_r$ and $e'_r$ such that $\fdelta(\flambda(u))\in e_r$ 
 and $\flambda(p)\in e'_r$ in $O(\log n)$ time using ray-shooting queries.
Similar to case (1), we compute $\Theta_{Z_i}$
for each \dsc $Z_1=\{e_1,e_l,e_2\}$, $Z_2= \{e_1,e_l,e_r,e_2\}$ and $Z_3= \{e_1,e_l,e'_r,e_2\}$
and check if $\gth(Z_i)\subseteq P$ for $\theta \in \Theta_{Z_i}$.
Then we compute the \bc{s} satisfying $Z_i$ with $\cornerc$ on an end vertex
of an edge in $Z_i$ and check their feasibility.

\begin{lemma}
\label{lem:F}
Our algorithm computes all maximal configurations of type \typef
and all breaking configurations of type \typef which contain a \cornerc on a vertex of $P$.
\end{lemma}
\begin{proof}
Let $Z_1=\{e_1,e_l,e_2\}$ and $Z_2=\{e_1,e_l,e_r,e_2\}$ be a \dsc of type $\typef_1$ and $\typef_2$, 
respectively, and assume they are feasible at some $\theta$.
There are three events corresponding to $Z_1$, $(e_1,e_2)$ becomes \textsf{h}-misaligned at $\theta_1$,
$(e_1,e_l)$ becomes \textsf{v}-misaligned at $\theta_2$ and $(e_1,e_l,e_2)$ becomes \textsf{t}-misaligned at $\theta_3$.
Let $\theta'$ be the smallest one among $\theta_1$, $\theta_2$ and $\theta_3$.
Then it is easy to see that $Z_1$ is found at the event corresponding to $\theta'$.
For example, if $\theta' = \theta_1$, when $u$ and $v$ are aligned horizontally, there exists exactly one point $t$ on $e_2$ that realizes the \textsf{h}-alignment of $e_1$ and $e_2$, and $\feta(t)$ meets $u$ and $e_1$. Thus there exists at most one edge that is \textsf{t}-aligned with $e_1$ and $e_2$, which should be $e_l$ by \textsf{t}-alignment of $(e_1, e_l, e_2)$. Moreover, $\fdelta(\feta(u))=\fdelta(\feta(t))$ is on $e_l$, which allows us to detect $(e_1, e_l, e_2)$ at $\theta'$.
Similarly, there are eight events corresponding to $Z_2$, four events are of case (1) and others are of case (2).
Then $Z_2$ is found at the smallest orientation among the orientations that the eight events occur
since all incident pairs and triplets of edges in $Z_2$ are \textsf{h}-, \textsf{v}- or \textsf{t}-aligned 
until the smallest orientation.
Thus we find all \dsc{s} of type $\typef$.
Since we compute the maximal configuration and breaking configuration which contain 
a \cornerc on a vertex of $P$
whenever we find a \dsc, all maximal configurations 
and breaking configurations of type \typef which contains a \cornerc on a vertex of $P$ 
are computed.
\end{proof}
There are $O(n^2)$ events corresponding to case (1) and 
$O(n^3)$ events corresponding to case (2). We can compute them in $O(n^3)$ time in total.
For each event, we find $O(1)$ \dsc{s} in $O(\log n)$ time and compute $O(1)$ maximal and breaking configurations of each \dsc in $O(1)$ time, and check their feasibility in $O(\log n)$ time.
And the only data structure we use for type $\typef$ is a ray-shooting data structure of $O(n)$ space. 
By lemmas~\ref{lem:F_lmr} and~\ref{lem:F}, we can conclude with the following lemma.

\begin{lemma}
We can compute a largest rectangle among all \lmr{s} 
of type \typef in $O(n^3\log n)$ time using $O(n)$ space.
\end{lemma}

\section{Computing a largest rectangle in a simple polygon with holes}
\label{sec:with-holes}
Our algorithm can compute a largest rectangle in a simple polygon $P$ 
with $h$ holes and $n$ vertices.
We use the same classification of largest rectangles 
and find the \lmrs of the six types. 
We construct a ray-shooting data structure, such as the one by
Chen and Wang~\cite{chen2015visibility} in $O(n+h^2\polylog h)$ time using $O(n+h^2)$ space,
which supports a ray-shooting query in $O(\log n)$ time.
We also construct the visibility region from each vertex of $P$,  
which can be done in $O(n^2\log n)$ time using $O(n^2)$ space
by using the algorithm in~\cite{chen2015visibility}.
Each visibility region is simple and has $O(n)$ complexity.
The staircase of a vertex of $P$ can be constructed in $O(n\log n)$ time 
using plane sweep with ray-shooting queries. 
Each staircase of a vertex $u$ of $P$
has $O(n)$ space. There are $O(n^2)$ events to the staircase of $u$ 
since it is equivalent to the staircase constructed in $\vis(u)$, a simple polygon
with $O(n)$ vertices.

We say a rectangle is \textit{empty} if there is no hole contained in it. 
Since $P$ has holes, there can be a hole contained in a rectangle $R$
even though every side of $R$ is contained in $P$.
Thus, we check the emptiness of rectangles, together with the test
for their sides being contained in $P$.
The emptiness of a rectangle can be checked by constructing 
the triangular range searching data structure proposed 
by Goswami~et~al.~\cite{goswami2004triangular}.
Given a set of $n$ points in the plane, the triangular range searching data 
structure can be constructed in $O(n^2)$ time and space
such that
given a query triangle, the number of points lying in the triangle can be
answered in $O(\log n)$ time.
We consider the $n$ vertices of $P$ as the input points to the data structure 
and consider a rectangle as a query with two triangles obtained from subdividing 
the rectangle by a diagonal.
Then we can check the emptiness of a rectangle in $O(\log n)$ time.
Since the remaining part of our algorithm works as it is,
we have Theorem~\ref{thm:final}.
\section{Computing a largest rectangle in a convex polygon}
 When $P$ is convex, there is no reflex vertex and therefore it suffices to consider only
 the \lmr{s} of types \typea and \typef. 
 For type \typea, we take two vertices
  $v$ and $u$ and check if the square with $vu$ as a diagonal is contained in $P$
  using ray-shooting queries.
  Among such squares, the one with largest area is
  the largest \lmr of type \typea in $P$, by Lemma \ref{lem:type-A-square}.
   Thus, using the method in Lemma \ref{lem:lmrs-type-A}, we can
  compute a largest \lmr of type \typea in $O(n^2 \log n)$
  time using $O(n)$ space.
  
   For type \typef, we find the events considered in Section~\ref{sec:comp-F}
 and all \dsc{s} corresponding to the events in case (1) 
 that a vertex $u$ is aligned to another vertex in $O(n)$ time by maintaining rays 
 $\lambda(u)$, $\delta(\flambda(u))$, 
 $\delta(u)$ and $\lambda(\fdelta(u))$ during the rotation.
 Since $P$ is convex, the foot of each ray emanating from $u$ changes \textit{continuously} along the boundary of $P$.
 Similarly, we find all \dsc{s} corresponding to the events in case (2) in Section~\ref{sec:comp-F}  that 
 an edge triplet becomes \textsf{t}-misaligned. It is caused by $\fdelta(\feta(u))$ meeting $p$ for 
 a vertex pair $(u,p)$ and its corresponding \dsc{s} can be computed in $O(n)$ time 
 by maintaining $\eta(u)$, $\delta(u)$, $\xi(p)$ and $\lambda(p)$
 during the rotation, where $\xi(p)$ is the vertically upward ray from $p$.
 Thus, we can find all events and their corresponding \dsc{s} in $O(n^2)$ time for case (1) and
 in $O(n^3)$ time for case (2). Since every \lmr is contained in $P$,
 we can find the maximal configuration of each \dsc in $O(1)$ time.
 We conclude with Theorem~\ref{thm:convex}.

\section{Discussion}
We study  the problem of finding a maximum-area rectangle with no restriction on
its orientation that is contained in a simple polygon $P$ with $n$ vertices in the plane and
present $O(n^3\log n)$-time algorithm. One may wonder if the algorithm can be
improved. It is shown that there can be $\Omega(n^3)$ combinatorially distinct
rectangles even when $P$ is convex~\cite{Cabello-2016}. So any algorithm that iterates
over all combinatorially distinct rectangles contained in $P$ needs at least
$\Omega(n^3)$ time, and therefore there seems hardly any room to improve, except
improving it by $\log n$ factor.

\bibliography{papers}{}
\begin{figure}[ht]
  \centering
\includegraphics[width=.8\textwidth]{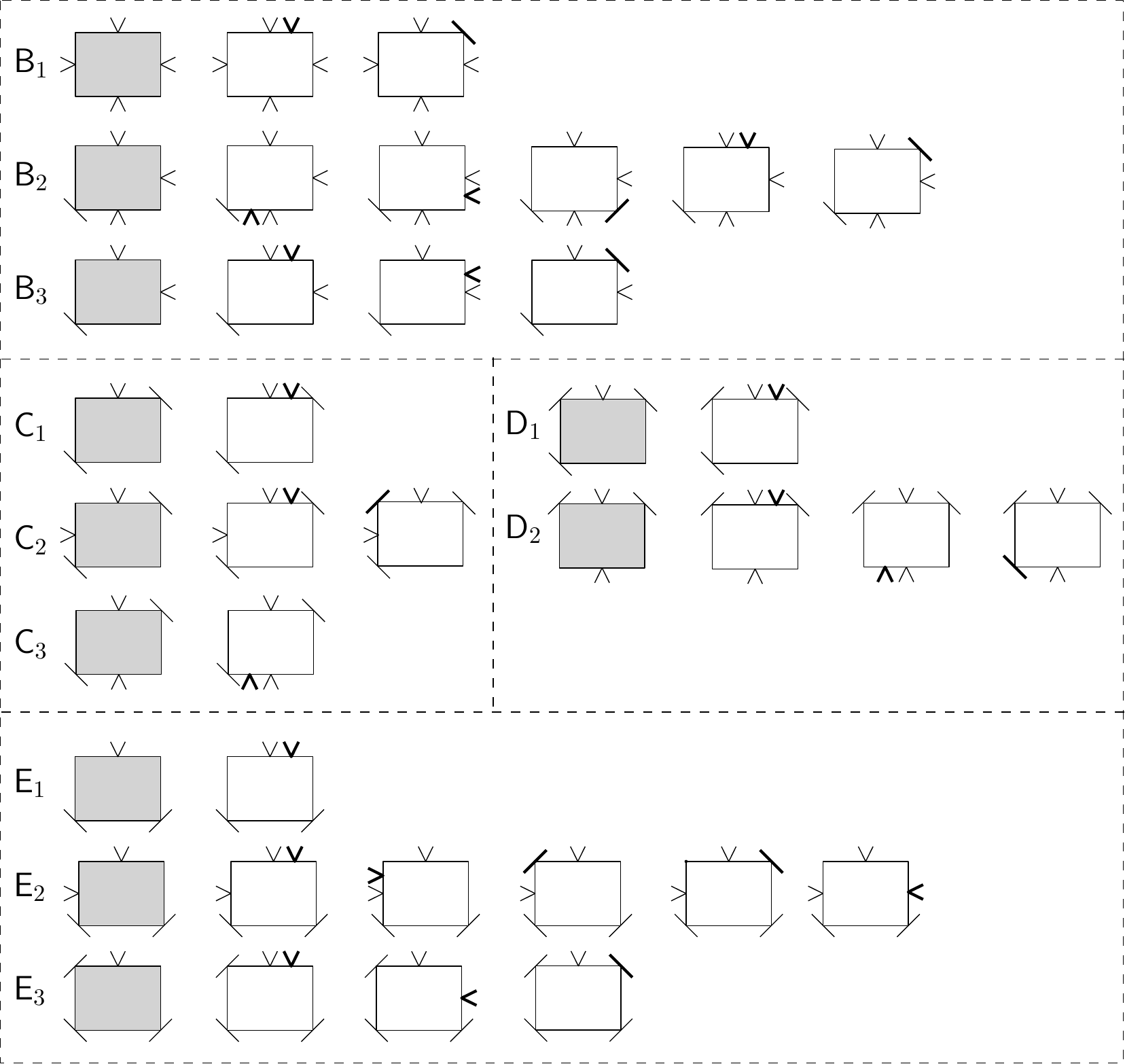}
\caption{Canonical (sub)types (gray rectangles) and their breaking configurations without duplication. The breaking configurations of subtypes $\typef_1$ and $\typef_2$ appear as breaking configurations of other types: By adding a \sidec to a \dsc $Z$ of type $\typef_1$, $Z$ becomes a \bc of type either $\typed_1$ or $\typee_3$. By adding a \cornerc to a \dsc $Z$ of type $\typef_1$, $Z$ becomes a \bc of type $\typef_2$. By adding a \sidec to a \dsc $Z$ of type $\typef_2$, $Z$ becomes a \bc (of the last type) of type $\typee_3$.}
\label{fig:BCs}
\end{figure}
\end{document}